%% file: main.tex
\documentclass{article}
\usepackage[papersize={8.5in,11in},margin=1in]{geometry}

\usepackage{booktabs} 

\usepackage[utf8]{inputenc} 
\usepackage[T1]{fontenc}    
\usepackage{hyperref}       
\usepackage{url}            
\usepackage{booktabs}       
\usepackage{amsfonts}       
\usepackage{nicefrac}       
\usepackage{microtype}      

\usepackage{graphicx}
\usepackage{subfigure}
\usepackage{caption}
\usepackage{xcolor}
\usepackage{libertine}
\usepackage{newtxmath}
\usepackage{natbib}
\usepackage{authblk}
\usepackage{balance}

\input{header.tex}

\newcommand{\tsty}{\displaystyle}

\title{An economic approach to vehicle dispatching for ride sharing}

\author{Mengjing Chen}
\author{Weiran Shen}
\author{Pingzhong Tang}
\author{Song Zuo}

\affil{IIIS, Tsinghua University
\thanks{Contacts: \texttt{ccchmj@qq.com,
\{emersonswr, kenshinping, songzuo.z\}@gmail.com}}}

\begin{document}

\maketitle


\begin{abstract}
\input{abstract.tex}
\end{abstract}

%
%


\input{intro.tex}

\input{model.tex}

\input{formulation.tex}

\input{static.tex}
\input{dual.tex}
\input{experiment.tex}


\bibliographystyle{ACM-Reference-Format}
\bibliography{ref}

\end{document}

%% file: header.tex
\usepackage{framed}

\usepackage{amsthm}
\usepackage{amsmath}
\usepackage{amssymb}
\usepackage{xspace}
\usepackage{booktabs}
\usepackage{aliascnt}
\usepackage{xcolor}
\usepackage{balance}
\usepackage{enumitem}

\makeatletter
\renewcommand{\paragraph}{%
  \@startsection{paragraph}{4}%
  {\z@}{1.0ex \@plus 0.5ex \@minus .2ex}{-1em}%
  {\normalfont\normalsize\it}%
}
\makeatother

\usepackage{fancyhdr} 
\usepackage{xcolor} 
\definecolor{note_fontcolor}{rgb}{0.800781, 0.800781, 0.800781}
\definecolor{ForestGreen}{rgb}{0.1333,0.5451,0.1333}

\usepackage{changepage}
\usepackage[nofillcomment,linesnumbered,noend,noresetcount,noline]{algorithm2e}

\usepackage{amssymb} 
\usepackage{graphicx} 
\usepackage{enumitem} 
\usepackage{shorttoc} 
\usepackage{array} 
\usepackage{float}
\usepackage{verbatim}

\newtheorem{theorem}{Theorem}[section]

\newaliascnt{lemma}{theorem}
\newaliascnt{assumption}{theorem}
\newaliascnt{proposition}{theorem}
\newaliascnt{corollary}{theorem}
\newaliascnt{claim}{theorem}
\newaliascnt{observation}{theorem}
\newaliascnt{definition}{theorem}
\newaliascnt{fact}{theorem}
\newaliascnt{statement}{theorem}
\newaliascnt{mechanism}{theorem}
\newaliascnt{example}{theorem}
\newaliascnt{remark}{theorem}

\newtheorem{Lemma}[lemma]{Lemma}

\newtheorem{observation}[observation]{Observation}

\newtheorem{definition}[definition]{Definition}

\aliascntresetthe{lemma}
\aliascntresetthe{assumption}
\aliascntresetthe{proposition}
\aliascntresetthe{corollary}
\aliascntresetthe{claim}
\aliascntresetthe{observation}
\aliascntresetthe{definition}
\aliascntresetthe{fact}
\aliascntresetthe{statement}
\aliascntresetthe{mechanism}
\aliascntresetthe{example}
\aliascntresetthe{remark}

\theoremstyle{remark}

\numberwithin{equation}{section}

\newtheoremstyle{etplain}
{.0\baselineskip\@plus.1\baselineskip\@minus.1\baselineskip}
{.0\baselineskip\@plus.1\baselineskip\@minus.1\baselineskip}
{\itshape}
{}
{\bfseries}
{.}
{ }
{}

\newcommand{\namedref}[2]{\hyperref[#2]{#1~\ref*{#2}}}

\newcommand{\figurerefb}[2]{\hyperref[#1]{Figure~\ref*{#1}#2}}

\newcommand{\equationref}[1]{\hyperref[#1]{(\ref*{#1})}}
\renewcommand{\eqref}{\equationref}

\usepackage[textsize=tiny]{todonotes}

\newcommand{\DEBUG}[1]{}

\newcommand{\argmax}{\operatornamewithlimits{arg\,max}}

\newcommand\R{\mathbb{R}}

\DeclareMathOperator*{\E}{\mathbb E}

\renewcommand\comment[1]{}

\newcommand{\vin}{\mathrm{IN}}
\newcommand{\vout}{\mathrm{OUT}}

\newcommand{\obj}{\mathrm{OBJ}}

\newcommand{\fixed}{\texttt{FIXED}\xspace}
\newcommand{\surge}{\texttt{SURGE}\xspace}
\newcommand{\dynam}{\texttt{DYNAM}\xspace}

\usepackage[normalem]{ulem}

\newcommand{\hide}[1]{}
\newcommand{\nocolor}[1]{}
\newcommand{\replace}[2]{\hide{\color{red} \sout{#1}}{\xspace\nocolor{\color{blue}} #2}}
\newcommand{\N}{\mathbb{N}}


%% file: abstract.tex
Over the past few years, ride-sharing has emerged as an effective way to relieve
traffic congestion. A key problem for these platforms is to come up with a
revenue-optimal (or GMV-optimal) pricing scheme and an induced vehicle
dispatching policy that incorporate geographic and temporal information. In this
paper, we aim to tackle this problem via an economic approach.

Modeled naively, the underlying optimization problem may be non-convex and thus
hard to compute. To this end, we use a so-called ``ironing'' technique to
convert the problem into an equivalent convex optimization one via a clean
Markov decision process (MDP) formulation, where the states are the driver
distributions and the decision variables are the prices for each pair of
locations. Our main finding is an efficient algorithm that computes the exact
revenue-optimal (or GMV-optimal) randomized pricing schemes. We characterize the
optimal solution of the MDP by a primal-dual analysis of a corresponding convex
program. We also conduct empirical evaluations of our solution through real data
of a major ride-sharing platform and show its advantages over fixed pricing
schemes as well as several prevalent surge-based pricing schemes.


%% file: intro.tex
\section{Introduction}\label{sec:intro}

  The recently established applications of shared mobility, such as ride-sharing,
  bike-sharing, and car-sharing, have been proven to be an effective
  way to utilize redundant transportation resources and to optimize social
  efficiency~\citep{cramer2016disruptive}. Over the past few years, intensive
  researches have been done on topics related to the economic aspects of shared
  mobility~\citep{crawford2011new,kostiuk1990compensating,oettinger1999empirical}.

  Despite these researches, the problem of how to design revenue optimal prices
  and vehicle dispatching schemes has been
  largely open and one of the main research agendas in sharing economics. There
  are at least two challenges when one wants to tackle this problem in the
  real-world applications. First of all, due to the nature of transportation,
  the price and dispatch scheme must be geographically dependent. Secondly, the
  price and dispatch scheme must take into consideration the fact that supplies
  and demands in these environments may change over time. As a result, it may be
  difficult to compute, or even to represent a price and dispatch scheme for
  such complex environments.

  Traditional price and dispatch schemes for taxis
  \citep{laporte1992vehicle,gendreau1994tabu,ghiani2003real} and airplanes
  \citep{gale1993advance,stavins2001price,mcafee2006dynamic} do not capture the
  dynamic aspects of the environments: taxi fees are normally calculated by a
  fixed rate of distance and time and the
  prices of flight tickets are sold via relatively long booking periods, while
  in contrast, the customers of shared vehicles make their decisions instantly.

  The dynamic ride-sharing market studied in this paper is also known to have
  imbalanced supply and demand, either globally in a city or locally in a
  particular time and location. Such imbalance in supply and demand is known to
  cause severe consequences on revenues (e.g, the so-called {\em wild goose
  chase phenomenon}~\citep{castillo2017surge}). Surging price is a way to balance
  dynamic supply and demand \citep{chen2015dynamic} but there is no known
  guarantee that surge based pricing can dispatch vehicles efficiently and solve
  the imbalanced supplies and demands. Traditional dispatch schemes
  \citep{laporte1992vehicle,gendreau1994tabu,ghiani2003real} focus more on the
  algorithmic aspect of static vehicle routing, without consider pricing.
  However, vehicle dispatching and pricing problem are tightly related, since a
  new price scheme will surely induces a change on supply and demand since the
  drivers and passengers are strategic. In this paper, we aim to come up with
  price schemes with desirable induced supplies and demands.

  \subsection{Our contribution}

  In this paper, we propose a graph model to analyze the vehicle pricing and
  dispatching problem mentioned above. In the graph, each node refers to a
  region in the city and each edge refers to a possible trip that includes a
  pair of origin and destination as well as a cost associated with the trip on
  this edge. The design problem is, for the platform, to set a price and specify
  the vehicle dispatch for each edge at each time step. Drivers are considered
  to be non-strategic in our model, meaning that they will accept whatever offer
  assigned to them. The objective of the platform can either be its revenue or
  the GMV or any convex combination between them.

  Our model naturally induces a {\em Markov Decision Process} (MDP) with the
  driver distributions on each node as states, the price and dispatch along each
  edge as actions, and the revenue as immediately reward. Although the
  corresponding mathematical program is not convex (thus computationally hard to
  compute) in general, we show that it can be reduced to a convex one without
  loss of generality. In particular, in the resulting convex program where the
  throughput along each source and destination pair in each time period are the
  variables, all the constraints are linear and hence the exact optimal
  solutions can be efficiently computed (\autoref{thm:mdp}).

  We further characterize the optimal solution via primal-dual analysis. In
  particular, a pricing scheme is optimal if and only if the marginal
  contribution of the throughput along each edge equals to the system-wise
  marginal contribution of additional supply minus the difference of the long
  term contributions of unit supply at the origin and the destination (see
  \autoref{sec:dual}).

  We also perform extensive empirical analysis based on a public dataset
  with more than $8.5$ million orders.
  We compare our policy with other intensively studied policies such as surge
  pricing \citep{chen2015dynamic,cachon2016role,castillo2017surge}. Our
  simulations show that, in both the static and the dynamic environment, our
  optimal pricing and dispatching scheme outperforms surge pricing by $17\%$ and
  $33\%$. Interestingly, our simulations show that our optimal policy has much
  stronger ability in dispatching the vehicles than other policies, which
  results directly in its performance boost (see \autoref{sec:experiment}).


  \input{relatedwork.tex}

%% file: relatedwork.tex
\subsection{Related work}
Driven by real-life applications, a large number of researches have been done
on ride-share markets. Some of them employ queuing networks to model
the markets \citep{iglesias2016bcmp,banerjee2015pricing,tang2016coordinating}.
\citet{iglesias2016bcmp} describe the market as a
closed, multi-class BCMP queuing network which captures the randomness of
customer arrivals. They assume that the number of customers is fixed, since
customers only change their locations but don't leave the network. In contrast,
the number of customer are dynamic in our model and we only consider the one
who asks for a ride (or sends a request to the platform).
\citet{banerjee2015pricing} also use a queuing
theoretic approach to analyze the ride-share markets and mainly focus on the
behaviors of drivers and customers. They assume that the drivers enter or leave the
market with certain possibilities.
\citet{bimpikis2016spatial} take account for the
spatial dimension of pricing schemes in ride-share markets. They price for
each region and their goal is to rebalance the supply and demand of the whole
market. However, we price for each routing and aim to maximize the total revenue or
social welfare of the platform. We also refer the readers to the line of
researches initiated by \citep{ma2013t} for the problems about the car-pooling in
the ride-sharing systems \citep{alonso2017demand,zhao2014incentives,chan2012ridesharing}.

Many works on ride-sharing consider both the customers and the drivers to be
strategic, where the drivers may reject the requests or leave the system if the
prices are too low \citep{banerjee2015pricing,fang2017prices}. As we mentioned,
if the revenue sharing ratios between the platform and the drivers can be
dynamic, then the pricing problem and the revenue sharing problem could be
independent and hence the drivers are non-strategic in the pricing problem.
In addition, the platform can also increase the profit by adopting dynamic
revenue sharing schemes \citep{balseiro2017dynamic}.


Another work closely related to ours is by
\citet{banerjee2017pricing}. Their work is
concurrent and has been developed independently from ours. In particular, the
customers arrive according to a queuing model and their pricing policy is
state-independent and depends on the transition volume. Both their and our
models are built upon the underlying Markovian transitions between the states
(the distribution of drivers over the graph). The major differences are: (i) our
model is built for the dynamic environments with a very large number of
customers (each of them is non-atomic) to meet the practical situations, while
theirs adopts discrete agent settings; (ii) they overcome the non-convexity of
the problem by relaxation and focus only on concave objectives, which makes this work hard to use for real applications, while we solve
the problem via randomized pricing and transform the problem to a convex
program; (iii) they prove approximation bounds of the relaxation problem, while
we give exact optimal solutions of the problem by efficiently solving the convex
program.

%% file: model.tex
\section{Model}\label{sec:model}

A passenger (she) enters the ride-sharing platform and sends a request including
her origin and destination to the platform. The platform receives the
request and determines a price for it. If user accepts the price, \replace{a
driver will be dispatched to pick her up, otherwise not}{then the platform may
decide whether to send a driver (he) to pick her up}. The platform is also able
to dispatch drivers from one place to another even there is no request to be
served. By the pricing and dispatching methods above, the goal of maximizing
revenue or social welfare of the entire platform can be achieved. Our model
incorporates the two methods into a simple pricing problem. In this section, we
define basic components of our model and consider two settings:
dynamic environments with a finite time horizon and static environments with an
infinite time horizon. Finally we reduce the action space of the problem and give a simple formulation.

  \paragraph{\bf Requests}
	We use a strongly connected digraph $G = (V, E)$ to model the geographical
  information of a city. Passengers can only take rides from nodes to nodes on
  the graph. When a passenger enters the platform, she expects to get a ride from node $s$ to node $t$, and is willing to pay at most $x\geq 0$ for the ride. She then sends to the platform a request, which is \replace{a tuple $(e, x)$}{associated with the tuple
  $e = (s, t)$}.
Upon receiving the request, The platform sets a price $p$ for it. If the price is accepted by the passenger (i.e., $x\geq p$), then the platform tries to send a driver to pick her up. We say that the platform rejects the request, if no driver is available.

A request is said to be \emph{accepted} if both the passenger accepts the price $p$ and there are available drivers. Otherwise, the request is considered to end immediately.

  \paragraph{\bf Drivers}
  \replace{A driver is available at node $v$ if he can provide a ride for a
  passenger whose ride starts form $v$. For any $v \in V$, let $w(v)$ be the
  number of available drivers at node.}{
  Clearly, within each time period,
  the total number of accepted requests starting from $s$ cannot be more than
  the number of drivers available at $s$. Formally, let $q(e)$ denote the total
  number of accepted request along edge $e$, then:
  \begin{align}\label{eq:feasible}
    \tsty
    \sum_{e \in \vout(v)} q(e) \leq w(v),{~}\forall v \in V,
  \end{align}
  where $\vout(v)$ is the set of edges starting from $v$ and $w(v)$ is the
  number of currently available drivers at node $v$.}

  In particular, we assume that both the total number of drivers and the number
  of requests are very large, which is often the case in practice, and consider
  each driver and each request to be {\em non-atomic}. For simplicity, we
  normalize the total amount of drivers on the graph to be $1$, thus $w(v)$ is a
  real number in $[0, 1]$. We also normalize the number of requests on each edge
  with the total number of drivers. Note that the amount of requests on an edge
  $e$ can be more than $1$, if there are more requests on $e$ than the total
  drivers on the graph.

  \replace{}{
  \paragraph{\bf Geographic Status}
  For each accepted request on edge $e$, the platform will have to cover a
  transportation cost $c_\tau(e)$ for the driver. In the meanwhile, the assigned
  driver, who currently at node $s$, will not be available until he arrives the
  destination $t$. Let $\Delta \tau(e)$ be the traveling time from $s$ to $t$
  and $\tau$ be the timestep of the driver leaving $s$. He will be available
  again at timestep $\tau + \Delta \tau(e)$ on node $t$. Formally, the amount of
  available drivers on any $v \in V$ is evolving according to the following
  equations:
  \begin{align}\label{eq:trans}
    w_{\tau + 1}(v)
      = w_\tau(v) - \sum_{e \in \vout(v)} q_\tau(e)
        + \sum_{e \in \vin(v)} q_{\tau + 1 - \Delta \tau(e)}(e),
  \end{align}
  where $\vin(v)$ is the set of edges ending at $v$. Here we add subscripts to
  emphasize the timestamp for each quantity. In particular, throughout this
  paper, we focus on the discrete time step setting, i.e., $\tau \in \N$.
  }

  \replace{}{
  \paragraph{\bf Demand Function}
  As we mentioned, the platform could set different prices for the requests.
  Such prices may vary with the request edge $e$, time step $t$, and the driver
  distribution but must be independent of the passenger's private value $x$ as
  it is not observable.
  Formally, let $D_\tau(\cdot | e): \R_+ \rightarrow \R_+$ be the {\em
  demand function} of edge $e$, i.e., $D_\tau(p | e)$ is the amount of requests
  on edge $e$ with private value $x \geq p$ in time step $\tau$.\footnote{In
  practice, such a demand function can be predicted from historical data
  \cite{tong2017simpler,moreira2013predicting}.} Then the amount
  of accepted requests $q_\tau(e) \leq \E[D_\tau(p_\tau(e) | e)]$, where the
  expectation is taken over the potential randomness of the pricing rule
  $p_\tau(e)$.\footnote{The randomized pricing rule may set different prices for
  the requests on the same edge $e$.}

  \paragraph{\bf Design Objectives}
  In this paper, we consider a class of \emph{state-irrelevant} objective
  functions. A function is state-irrelevant if its value only depends on the
  amount of accepted request on each edge $q(e)$ but not the driver distribution
  of the system $w(v)$. Note that a wide range of objectives are included in our
  class of objectives, such as the revenue of the platform:
	\begin{gather*}
	  \mathrm{REVENUE}(p, q)
      = \sum_{e, \tau}\E[(p_\tau(e) - c_\tau(e))\cdot q_\tau(e)],
	\end{gather*}
	and the social welfare of the entire system:
	\begin{gather*}
	  \mathrm{WELFARE}(p, q)
      = \sum_{e, \tau}\E[(x - c_\tau(e)) \cdot q_\tau(e)].
	\end{gather*}

	In general, our techniques work for any state-irrelevant objectives. Let
  $g(p, q)$ denote the general objective function and the dispatching and
  pricing problem can be formulated as follows:
  \begin{align}\label{eq:program}
    \text{maximize}   \quad & \tsty \sum_{e, \tau} g(p_\tau(e), q_\tau(e) | e)
    \\
    \text{subject to} \quad & \eqref{eq:feasible} ~\text{and}~ \eqref{eq:trans}.
    \nonumber
  \end{align}

  \paragraph{\bf Static and Dynamic Environment}
  In general, our model is defined for a dynamic environment in the sense that
  the demand function $D^\tau$ and the transportation cost $c_\tau$ could be
  different for each time step $\tau$. In particular, we study the problem
  \eqref{eq:program} in general dynamic environments with finite time horizon
  from $\tau = 1$ to $T$, where the initial driver distribution $w_1(v)$ is
  given as input.

  In addition, we also study the special case with {\em static environment} and
  infinite time horizon, where $D^\tau \equiv D$ and $c_\tau \equiv c$ are
  consistent across each time step.
  }

  \subsection{Reducing the action space}

  In this section, we rewrite the problem to an equivalent reduced
  form by incorporating the action of dispatching into pricing, i.e., using $p$
  to express $q$. The idea is straightforward: (i) for the requests rejected by
  the platform, the platform could equivalently set an infinitely large price;
  (ii) if the platform is dispatching available drivers (without requests) from
  node $s$ to $t$, we can create virtual requests from $s$ to $t$ with $0$ value
  and let the platform sets price $0$ for these virtual requests. In fact, we
  can assume without loss of generality that $D(0 | e) \equiv 1$, the total
  amount of drivers, because one can always add enough virtual requests for the
  edges with maximum demand less than $1$ or remove the requests with low values
  for the edges with maximum demand exceeds the total driver supply, $1$.

  As a result, we may conclude that $q(e) \leq D(p | e)$. Since our goal is to
  maximize the objective $g(p, q)$, raising prices to achieve the same amount of
  flow $q(e)$ (such that $\E[D(p | e)] = q(e)$) never eliminates the optimal
  solution. In other words,
  \begin{observation}\label{obs:reduce}
    The original problem is equivalent to the following reduced problem, where
    the flow variables $q_\tau(e)$ are uniquely determined by the price
    variables $p_\tau(e)$:
    \begin{equation}\label{eq:formal}
    \begin{aligned}
      \text{maximize}   \quad &
          \tsty \sum_{e, \tau} g(p_\tau(e), D_\tau(p_\tau(e) | e))  \\
      \text{subject to} \quad &
          q_\tau(e) = \E[D(p_\tau(e) | e)]  \\
          & \eqref{eq:feasible} ~\text{and}~ \eqref{eq:trans}.
    \end{aligned}
    \end{equation}
  \end{observation}

%% file: formulation.tex
\section{Problem Analysis}\label{sec:form}


\replace{
In this section, we first show our model is flexible enough to extend to
practical situations. We present techniques to unify models and show how to
dispatch vehicles efficiently. Secondly, we reformulate the pricing problem to a
flow problem. Once we decide the number of served requests for each route, the
prices of requests are determined. Then we concentrate on \emph{state-irrelevant}
objective functions and tackle the non-convexity problem by randomized pricing
scheme and ''ironing'' technique. Finally, we formally define the dispatching
problem as a Markov decision process.
}{
In this section, we demonstrate how the original problem \eqref{eq:formal} can
be equivalently rewritten as a Markov decision process with a convex objective
function. Formally,

\begin{theorem}\label{thm:mdp}
  The original problem \eqref{eq:formal} of the instance $\langle G, D, g,
  \Delta \tau \rangle$ is equivalent to a Markov decision process problem of
  another instance $\langle G', D', g', \Delta \tau' \rangle$ with $g'$ being
  convex.
\end{theorem}

The proof of Theorem \ref{thm:mdp} will be immediate after Lemma \ref{lem:traveltime} and \ref{lem:flow}. The equivalent Markov decision process problem could be
formulated as a convex program, and hence can be solved efficiently.

\subsection{Unifying travel time}
Note that the original problem \eqref{eq:formal}, in general, is not a MDP by
itself, because the current state $w^{\tau + 1}(v)$ may depend on the action
$q^{\tau + 1 - \Delta \tau(e)}$ in \eqref{eq:trans}. Hence our first step is to
equivalently map the original instance to another instance with traveling time
is always $1$, i.e., $\Delta \tau(e) \equiv 1$:

\begin{Lemma}[Unifying travel time]\label{lem:traveltime}
  The original problem \eqref{eq:formal} of an general instance $\langle G, D,
  g, \Delta \tau\rangle$ is equivalent to the problem of a $1$-travel time
  instance $\langle G', D', g', \Delta \tau' \rangle$, where $\Delta
  \tau'(\cdot) \equiv 1$.
\end{Lemma}

\replace{Note that not all drivers currently in the system can provide ride services. A
working driver cannot handle another request before finishing the current one.
And different requests may require different time to finish. In order to
distinguish between working drivers and idle (can provide service immediately)
drivers, a possible solution is to add an indicator variable to each driver.
However,}{
Intuitively,} we tackle this problem by adding virtual nodes into the graph to
replace the original edges. This operation splits the entire trip into smaller
ones, and at each time step, all drivers become \replace{idle}{available}.

}


\begin{proof}
  For edges with traveling time $\Delta \tau(e) = 1$, we are done.

  For edges with traveling time $\Delta \tau(e) > 1$,
%
we add $\Delta \tau(e) - 1$ virtual nodes into the graph, i.e., $v^e_1, \ldots,
v^e_{\Delta \tau(e) - 1}$, and the directed edges connecting them to replace the
original edge $e$, i.e.,
\begin{gather*}
  \tsty
  E'(e) =
    \{(s, v^e_1), (v^e_1, v^e_2), \ldots, (v^e_{\Delta \tau(e) - 2},
      v^e_{\Delta \tau(e) - 1}), (v^e_{\Delta \tau(e) - 1}, t)\},  \\
  E' = \bigcup_{e \in E} E'(e), \quad V' = \bigcup_{e \in E} \{v^e_1, \ldots,
        v^e_{\Delta \tau(e) - 1}\} \cup V.
\end{gather*}
We set the demand function of each new edge $e' \in E'(e)$ to be identical to
those of the original edge $e$: $D'(\cdot | e') \equiv D(\cdot | e)$.

An important but natural constraint is that if a driver handles a request on
edge $e$ of the original graph, then he must go along all edges in $E'(e)$ of
the new graph, because he cannot leave the passenger halfway. To guarantee this,
we only need to guarantee that all edges in $E'(e)$ have the same price. Also,
we need to split the objective of traveling along $e$ into the new edges, i.e.,
each new edge has objective function
\begin{gather*}
  g'(p, q | e') = g(p, q | e) / \Delta \tau(e), \forall e'\in E'(e).
\end{gather*}
One can easily verify that the above operations increase the graph size to at
most $\max_{e \in E} \Delta \tau(e)^*$ times of that of the original one.
In particular, there is a straightforward bijection between the dispatching
behaviors of the original $G = (V, E)$ and the new graph $G' = (V', E')$. Hence
we can always recover the solution to the original problem.
%
\end{proof}




\subsection{Flow formulation and randomized pricing}

By \autoref{lem:traveltime}, the original problem \eqref{eq:formal} can be
formulated as an MDP:
\begin{definition}[Markov Decision Process]
  The vehicle pricing and dispatching problem is a Markov decision process,
  denoted by a tuple $(G, D, g, S, A, W)$, where $G = (V, E)$ is the given
  graph, $D$ is the demand function, objective $g$ is the reward function, $S =
  \Delta(V)$ is the state space including all possible driver distributions over
  the nodes, $A$ is the action space, and $W$ is the state transition rule:
  \begin{align}\label{eq:strans}
    \tsty
    w_{\tau + 1}(v) = w_\tau(v) - \sum_{e \in \vout(v)} q_\tau(e)
      + \sum_{e \in \vin(v)} q_\tau(e).
  \end{align}
\end{definition}

However, by na\"ively using the pricing functions $p_\tau(e)$ as the actions,
the induced flow $q_\tau(e) = \E[D_\tau(p_\tau(e) | e)]$, in general, is neither
convex nor concave. In other words, both the reward $g$ and the state transition
$W$ of the corresponding MDP is non-convex. As a result, it is hard to solve the
MDP efficiently.

In this section, we show that by formulating the MDP with the {\em flows}
$q_\tau(e)$ as actions, the corresponding MDP is convex.
\begin{Lemma}[Flow-based MDP]\label{lem:flow}
  In the MDP $(G, D, g, S, A, W)$ with all possible flows as the action set $A$,
  i.e., $A = [0, 1]^{|E|}$, the state transition rules are linear functions of
  the flows and the reward functions $g$ are convex functions of the flows.
\end{Lemma}

\begin{proof}
  To do this, we first need to rewrite the prices $p_\tau(e)$ as functions of
  the flows $q_\tau(e)$. In general, since the prices could be randomized, the
  inverse function of $q_\tau(e) = \E[D_\tau(p_\tau(e) | e)]$ is not unique.

  Note that conditional on fixed flows $q_\tau(e)$, the state transition of the
  MDP is also fixed. In this case, different prices yielding such specific flows
  only differs in the rewards. In other words, it is without loss of generality
  to let the inverse function of prices be as follows:
  \begin{align*}
    \tsty
    p_\tau(e) = \argmax_{p} g_\tau(p_\tau(e), q_\tau(e) | e), ~ \text{s.t.} ~
      q_\tau(e) = \E[D_\tau(p_\tau(e) | e)].
  \end{align*}

  In particular, since the objective function $g$ we studied in this paper is
  linear and weakly increasing in the prices $p$ and the demand function
  $D(p | e)$ is decreasing in $p$, the inversed price function could be defined
  as follows:
  \begin{itemize}
    \item Let $g_\tau(q | e) = g_\tau(D^{-1}_\tau(q | e), q | e)$, i.e., the
          objective obtained by setting the maximum fixed price $p =
          D^{-1}_\tau(q | e)$ such that the induced flow is exactly $q$;
    \item Let $\hat{g}_\tau(q|e)$ be the {\em ironed objective function}, i.e.,
          the smallest concave function that upper-bounds $g_\tau(q|e)$ (see
          \autoref{fig:iron});
    \item For any given $q_\tau(e)$, the maximum objective on edge $e$ is $\hat
          g_\tau(q_\tau(e) | e)$ and could be achieve by setting the price to be
          randomized over $D^{-1}_\tau(q' | e)$ and $D^{-1}_\tau(q'' | e)$.
  \end{itemize}
  \begin{figure}[h!]
  	\centering
  	\includegraphics[width=0.5\textwidth]{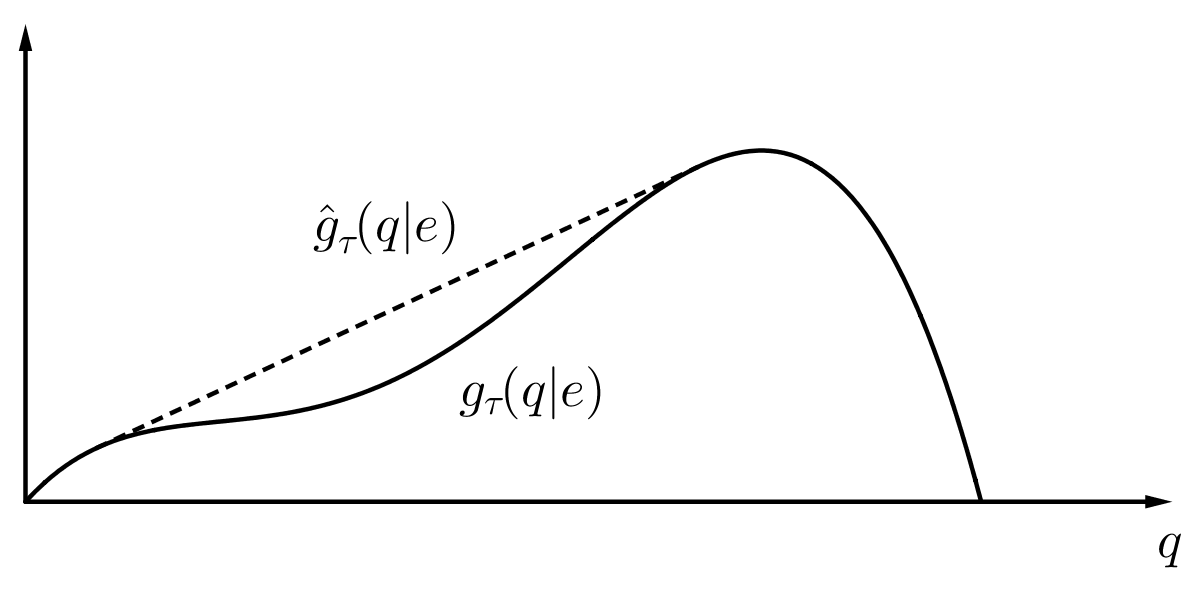}
  	\caption{Ironed objective function}
  	\label{fig:iron}
  \end{figure}
  Finally, we prove the above claim to complete the proof of \autoref{lem:flow}.

	By the definition of $\hat{g}_\tau(q|e)$, for any randomized price $p$,
	\begin{gather*}
    \tsty
		\E_{p}[g_\tau(D_\tau(p | e)|e)]
    \le \E_{p}[\hat{g}_\tau(D_\tau(p | e)|e)].
	\end{gather*}
	Since $\hat{g}$ is concave, applying Jensen's inequality yields:
	\begin{gather*}
    \tsty
		\E_{p}[\hat{g}_\tau(D_\tau(p | e)|e)]
      \le \hat{g}_\tau\left(\E_{p}[D_\tau(p | e)]\big|e\right)
      =\hat{g}_\tau(\bar{q}|e)
	\end{gather*}
	Now it suffices to show that the upper bound $\hat{g}_\tau(\bar{q}|e)$ is
  attainable.

	If $\hat{g}_\tau(\bar{q}|e)=g_\tau(\bar{q}|e)$, then the right-hand-side could
  be achieved by letting $p_\tau(e)$ be the deterministic price
  $D^{-1}_{\tau}(\bar q | e)$.

  Otherwise, let $I = (q', q'')$ be the ironed interval (where $\hat{g}_\tau(q|e)
  > g_\tau(q|e), \forall q \in I$ but $\hat{g}_\tau(q'|e)= g_\tau(q'|e)$ and
  $\hat{g}_\tau(q''|e)= g_\tau(q''|e)$) containing $\bar{q}$. Thus $\bar{q}$ can
  be written as a convex combination of the end points $q'$ and $q''$: $\bar q =
  \lambda q' + (1 - \lambda) q''$. Note that the function $\hat{g}_\tau$ is
  linear within the interval $I$. Therefore
	\begin{align*}
    \tsty
		\lambda g_\tau(q' | e) + (1 - \lambda)g_\tau(q'' | e)
      = \lambda \hat{g}_\tau(q' | e) + (1 - \lambda)\hat{g}_\tau(q'' | e)
			= \hat{g}_\tau(\lambda q' + (1 - \lambda)q'' | e)
         = \hat{g}_\tau(\bar{q} | e).
	\end{align*}
  In other words, the upper bound $\hat g_\tau(\bar{q} | e)$ could be achieved
  by setting the price to be $q'$ with probability $\lambda$ and $q''$ with
  probability $1 - \lambda$. In the meanwhile, the flow $q_\tau(e)$ would retain
  the same.
  %
\end{proof}

\begin{proof}[Proof of \autoref{thm:mdp}]
  The theorem is implied by \autoref{lem:traveltime} and \autoref{lem:flow}. In
  particular, the reward function is the ironed objective function $\hat g$.
\end{proof}

In the rest of the paper, we will focus on the following equivalent problem:
\begin{equation}\label{eq:convex}
\begin{aligned}
  \text{maximize}   \quad & \tsty \sum_{e, \tau} \hat g_\tau(q_\tau(e) | e)  \\
  \text{subject to} \quad & \eqref{eq:feasible} ~\text{and}~ \eqref{eq:strans}.
\end{aligned}
\end{equation}

%% file: static.tex
\section{Optimal Solution in Static Environment}
In this setting, we restrict our attention to the case where \replace{}{the
environment is static, hence} the objective function does not change over time,
i.e., $\forall \tau\in [T], \hat{g}_\tau(q|e) \equiv \hat{g}(q|e)$. We aim to
find the optimal stationary policy that maximizes the objective function, i.e.,
the decisions $q_\tau$ depends only on the current state $w_\tau$.

In this section, we discretize the MDP problem and focus on \emph{stable
policies}. With the introduction of the ironed objective function $\hat{g}_\tau$,
we show that for any discretization scheme, the optimal stationary policy of the
induced \emph{discretized MDP} is dominated by a stable dispatching scheme. Then
we formulate the stable dispatching scheme as a convex problem, which means the
optimal stationary policy can be found in polynomial time.
\begin{definition}
	A stable dispatching scheme is a pair of state and policy $(w_\tau, \pi)$,
  such that if policy $\pi$ is applied, the distribution of available drivers
  does not change over time, i.e., $w_{\tau+1}(v) = w_\tau(v)$.
  %
\end{definition}

In particular, under a stable dispatching scheme, the state transition rule
\eqref{eq:strans} is equivalent to the following form:
\begin{align}\label{eq:ststrans}
  \tsty
  \sum_{e \in \vout(v)} q(e) = \sum_{e \in \vin(v)} q(e).
\end{align}

\begin{definition}
	Let $\mathcal{M}=(G, D, \hat g, S, A, W)$ be the original MDP problem. A
  discretized MDP $\mathcal{DM}$ with respect to $\mathcal{M}$ is a tuple $(G_d,
  D_d, \hat g_d, S_d, A_d, W_d)$, where $G_d=G$, $D_d=D$, $\hat g_d=\hat g$,
  $W_d=W$, $S_d$ is a finite subset of $S$, and $A_d$ is a finite subset of $A$
  that contains all feasible transition flows between every two states in $S_d$.
\end{definition}
\begin{theorem}
	\label{thm:stable}
	Let $\mathcal{DM}$ and $\mathcal{M}$ be a discretized MDP and the corresponding
	original MDP. Let $\pi_d: S_d\to A_d$ be an optimal stationary policy of
	$\mathcal{DM}$. Then there exists a stable dispatching scheme $(w, \pi)$, such
	that the time-average objective of $\pi$ in $\mathcal{M}$ is no less than that
	of $\pi_d$ in $\mathcal{DM}$.
\end{theorem}
\begin{proof}
	Consider policy $\pi_d$ in $\mathcal{DM}$. Starting from any state in $S_d$
	with policy $\pi_d$, let $\{w_\tau\}_0^\infty$ be the subsequent state
  sequence. Since $\mathcal{DM}$ has finitely many states and policy $\pi_d$ is
  a stationary policy, there must be an integer $n$, such that $w_n=w_m$ for
  some $m<n$ and from time step $m$ on, the state sequence become a periodic
  sequence. Define
	\begin{gather*}
  \tsty
	\bar{w}=\frac{1}{n-m}\sum_{k=m}^{n-1}w_k, \quad
	\bar{q}=\frac{1}{n-m}\sum_{k=m}^{n-1}\pi_d(w_k)
	\end{gather*}
	Denote by $\pi_d(w_k|e)$ or $q_d(e)$ the flow at edge $e$ of the decision
  $\pi_d(w_k)$. Sum the transition equations for all the time steps $m\le k <n$,
  and we get:
	\begin{align*}
	    \tsty \sum_{k=m}^{n-1}w_{k+1}(v) - \sum_{k=m}^{n-1}w_{k}(v)
	  = \tsty \sum_{k=m}^{n-1}\left(\sum_{\vin(v)}\pi_d(w_k|e)\right)
      - \sum_{k=m}^{n-1}\left(\sum_{\vout(v)}\pi_d(w_k|e)\right)
	\end{align*}
	\begin{gather*}
    \tsty
	\bar{w}(v)=\bar{w}(v)-\left(\sum_{\vout(v)}\bar{q}(e)\right) +\left(\sum_{\vin(v)}\bar{q}(e)\right)
	\end{gather*}
	Also, policy $\pi_d$ is a valid policy, so $\forall v\in V$ and $\forall m\le k<n$:
	\begin{gather*}
    \tsty
	\sum_{\vout(v)}q_k(e) \leq w_k(v)
	\end{gather*}
	Summing over $k$, we have:
	\begin{gather*}
    \tsty
	\sum_{\vout(v)}\bar{q}(e)\le \bar{w}(v)
	\end{gather*}
	Now consider the original problem $\mathcal{M}$. Let $w=\bar{w}$ and $\pi$ be any stationary policy such that:
	\begin{itemize}
		\itemsep0em
		\item $\pi(w)=\bar{q}$;
		\item starting from any state $w'\ne w$, policy $\pi$ leads to state $w$ within finitely many steps.
	\end{itemize}
	Note that the second condition can be easily satisfied since the graph $G$ is strongly connected.

	With the above definitions, we know that $(w, \pi)$ is a stable dispatching scheme. Now we compare the objectives of the two policies $\pi_d$ and $\pi$. The time-average objective function is not sensitive about the first finitely many immediate objectives. And since the state sequences of both policies $\pi_d$ and $\pi$ are periodic, Their time-average objectives can be written as:
	\begin{gather*}
    \tsty
		\obj(\pi_d) =\frac{1}{n-m}\sum_{k=m}^{n-1}\sum_{e\in E}\hat{g}(q_d(e)|e)\\
    \tsty
		\obj(\pi) = \sum_{e\in E} \hat{g}(\bar{q}(e)|e)
	\end{gather*}
	By Jensen's inequality, we have:
	\begin{gather*}
	\obj(\pi_d) \tsty
		=\frac{1}{n-m}\sum_{k=m}^{n-1}\sum_{e\in E}\hat{g}(q_d(e)|e)
	\tsty
    \le \sum_{e\in E}\hat{g}\left[\left(\frac{1}{n-m}\sum_{k=m}^{n}q_d(e) \right)\big|e \right]
    \tsty = \sum_{e\in E} \hat{g}(\bar{q}(e)|e) = \obj(\pi)
	\end{gather*}
\end{proof}

With \autoref{thm:stable}, we know there exists a stable dispatching scheme that
dominates the  optimal stationary policy of the our discretized MDP. Thus we now
only focus on stable dispatching schemes. The problem of finding an optimal
stable dispatching scheme can be formulated as a convex program with linear
constraints:
\begin{equation}
\label{prog:static}  
\begin{aligned}
  \text{maximize}   \quad & \tsty \sum_{e \in E} \hat g(q | e)  \\
  \text{subject to} \quad & \eqref{eq:feasible} ~\text{and}~ \eqref{eq:ststrans}.
\end{aligned}
\end{equation}

Because $\hat{g}(q|e)$ is concave, the program is convex. Since all convex
programs can be solved in polynomial time, our algorithm for finding optimal
stationary policy of maximizing the objective functions is efficient.


%% file: dual.tex
\section{Characterization of optimality}\label{sec:dual}

In this section, we characterize the optimal solution via dual analysis. For the
ease of presentation, we consider Program \ref{prog:static} in the static environment
with infinite horizon. Our characterization directly extends to
the dynamic environment.


The Lagrangian is defined to be
\begin{align*}
	 L(q, \lambda, \mu)
	 = & \tsty
      - \sum_{e \in E} \hat{g}(q|e)
			+ \lambda\left(\sum_{e \in E}q(e) - 1\right)
	 \tsty
      + \sum_{v \in V}\mu_v
					\left(\sum_{\vout(v)}q(e) - \sum_{\vin(v)}q(e)\right) \\
	 =& \tsty
      - \lambda + \sum_{e \in E} \left[- \hat g(q|e)
			+ (\lambda + \mu_s - \mu_{t})q(e)\right],
\end{align*}
where $s$ and $t$ are the origin and destination of $e$, i.e., $e = (s, t)$, and $\lambda$ and $\mu$ are Lagrangian multipliers with $\lambda\ge 0$. Note that we implicitly transform program \ref{prog:static} to the standard form that minimizes the objective $-\sum_{e \in E}\hat{g}(q^*|e)$.

The Lagrangian dual function is
\begin{align*}
\tsty
h(\lambda, \mu)=\inf_q L(q,\lambda, \mu)
=\sum_{e \in E} \left[- \hat g(\tilde{q}|e)
+ (\lambda + \mu_s - \mu_{t})\tilde{q}(e)\right],
\end{align*}
where $\tilde{q}(e)$ is a function of $\lambda$ and $\mu$ such that $\lambda+\mu_s-\mu_t=\hat{g}'(\tilde{q}|e)$, where $\hat{g}'(\tilde{q}|e)$ is the derivative of the objective function with respect to flow $q$.
The dual program corresponding to Program \ref{prog:static} is
\begin{equation}\label{prog:static_dual}
\begin{aligned}
\text{maximize} \quad & h(\lambda, \mu) \\
\text{subject to}\quad & \lambda \ge 0
\end{aligned}
\end{equation}

According to the KKT conditions, we have the following characterization for optimal solutions.
\begin{theorem}
	\label{thm:dual}
	Let $q^*(e)$ be a feasible solution to the primal program \ref{prog:static} and $(\lambda^*, \mu^*)$ be a feasible solution to the dual program \ref{prog:static_dual}. Then both $q^*(e)$ and $(\lambda^*, \mu^*)$ are primal and dual optimal with $-\sum_{e \in E}\hat{g}(q^*|e)=h(\lambda^*, \mu^*)$, if and only if
	\begin{gather}
    \tsty
	\lambda^*\left(\sum_{e \in E} q^*(e)-1 \right) = 0  \label{eq:slackness}\\
	\hat g'(q^*|e) = \lambda^* + \mu_s^* - \mu_{t}^*, \forall v \in V \label{eq:gradient}
	\end{gather}
\end{theorem}
\begin{proof}
	According to the definition of $h(\lambda, \mu)$, we have $h(\lambda^*, \mu^*)=\inf_qL(q,\lambda^*, \mu^*)$. Since $\hat{g}(q|e)$ are concave functions, Equation \ref{eq:gradient} is equivalent to the fact that $q^*(e)$ minimizes the function $L(q,\lambda^*, \mu^*)$.
	\begin{align*}
	h(\lambda^*, \mu^*)=& \tsty \inf_qL(q,\lambda^*, \mu^*) \\
	=&L(q^*,\lambda^*, \mu^*) \\
	=&\tsty - \sum_{e \in E} \hat{g}(q^*|e)
	+ \lambda^*\left(\sum_{e \in E}q^*(e) - 1\right)
	\tsty + \sum_{v \in V}\mu^*_v
	\left(\sum_{\vout(v)}q^*(e) - \sum_{\vin(v)}q^*(e)\right) \\
	=&\tsty - \sum_{e \in E} \hat{g}(q^*|e),
	\end{align*}
	where the last equation uses the Equation \ref{eq:slackness} and the fact that $q^*(e)$ is feasible.
\end{proof}
Continuing with \autoref{thm:dual}, we will analyze the dual variables from the economics angle and some interesting insights into this problem for real applications.
\subsection{Economic interpretation}\label{ssec:econinter}
The dual variables have useful economic interpretations (see
\cite[Chapter 5.6]{boyd2004convex}). $\lambda^*$ is the system-wise marginal
contribution of the drivers (i.e. the increase in the objective function when a
small amount of drivers are added to the system). Note that by the complementary
slackness (Equation \ref{eq:slackness}), if $\lambda^* > 0$, the sum of the
total flow must be $1$, meaning that all drivers are busy, and more requests can
be accepted (hence increase revenue) if more drivers are added to the system.
Otherwise, there must be some idle drivers, and adding more drivers cannot
increase the revenue.

$\mu_v^*$ is the marginal contribution of the drivers at node $v$. If we allow
the outgoing flow from node $v$ to be slightly more than the incoming flow to
node $v$, then $\mu_v$ is the revenue gain from adding more drivers at node $v$.

\subsection{Insights for applications}\label{ssec:insight}
The way we formulate and solve the problem, in fact, naturally leads to two
interesting insights into this problem, which are potentially useful for real
applications.

\paragraph{1. Scalability} In our model, the size of the convex program
increases linearly in the number of edges, hence quadratically in the number of
regions. This could be one hidden feature that is potentially an obstacle to
real applications, where the number of regions in a city might be quite large.

A key observation to the issue is that any dispatching policy induced by a real
system is a feasible solution of our convex program and any improvement (maybe
via gradient descent) from such policy in fact leads to a better solution for
this system. In other words, it might be hard to find the exact optimal or
nearly optimal solutions, but it is easy to improve from the current state. Therefore, in practice, the
platform can keep running the optimization in background and apply the most
recent policy to gain more revenue (or achieve a higher value of some
other objectives).

\paragraph{2. Alternative solution} As suggested by the characterization and its
economic interpretation, instead of solving the convex programs directly, we
also have an alternative way to find the optimal policy by solving the dual program. The optimal policy can
be easily recovered from dual optimal solutions. In particular, according to the economic interpretation of
dual variables, we need to estimate the marginal contributions of drivers.

More importantly, the number of dual variables ($=$ the number of regions) is
much smaller than the number of primal variables ($=$ the number of edges $\approx$
square of the former). So solving the dual program may be more efficient when applied to real systems, and is also of independent interest of this paper.

%% file: experiment.tex
\section{Empirical Analysis}\label{sec:experiment}
We design experiments to demonstrate the good performance of our algorithms for real applications. In this section, we first describe the dataset and then introduce how to extract useful information for our model from the dataset. Two benchmark policies, \fixed and \surge, are compared with our pricing policy. The result analysis includes demand-supply balance and instantaneous revenue in both static and dynamic environments.
	\subsection{Dataset}

	We perform our empirical analysis based on a public dataset from a major ride-sharing
	company. The dataset includes the orders in a
	city for three consecutive weeks 
	and the total number of orders is more than $8.5$ million. An order is
	created when a passenger send a ride request to the platform.

	Each order consists of a unique order ID, a passenger ID, a driver ID, an
	origin, a destination, and an estimated price, and the timestamp when the order is created (see
	\autoref{table:row} for example). The driver ID might be empty if no
	driver was assigned to pick up the passenger. There are $66$ major regions
	of the city and the origins and destinations in the dataset are given as the
	region IDs. We say a request is related to a region if the region is either the origin or the destination of the request. And the popularity of a region is defined as the number of related requests. Since some of the regions in the dataset have very low popularity values, we only consider the most popular $21$ or $5$ regions in
	the two settings (see \autoref{ssec:static_env} and \autoref{ssec:dynamic_env}
	respectively for details). The related
	requests of the most popular $21$ (or $5$) regions cover about $90\%$ (or
	$50\%$)	of the total requests in the original dataset.

	For ease of presentation, we relabel the region IDs in descending order
	of their popularities (so region \#$1$ is the most popular region).
	\autoref{fig:connectivity} illustrates the frequencies of requests on
	different origin-destination pairs. From the figure, one can see that
	the frequency matrix is almost symmetric and the destination of a request is
	most likely to be in the same region as the origin.

  \begin{table}[h]
    \begin{center}
  		\begin{tabular}{| c | c | c | c | c | c | c |}
  			\hline
  			order & driver & user & origin & dest & price  & timestamp  \\
  			\hline
  			hash  & hash   & hash & hash   & hash & $37.5$ & 01-15 00:35:11  \\
  			\hline
  		\end{tabular}
  		\caption{An example of a row in the dataset, where ``hash'' stands for some
  		hash strings of the IDs that we didn't show the exact value here.}
  		\label{table:row}
  	\end{center}
  \end{table}

	\begin{figure}[h]
    \hfill
		\includegraphics[width=0.45\textwidth]{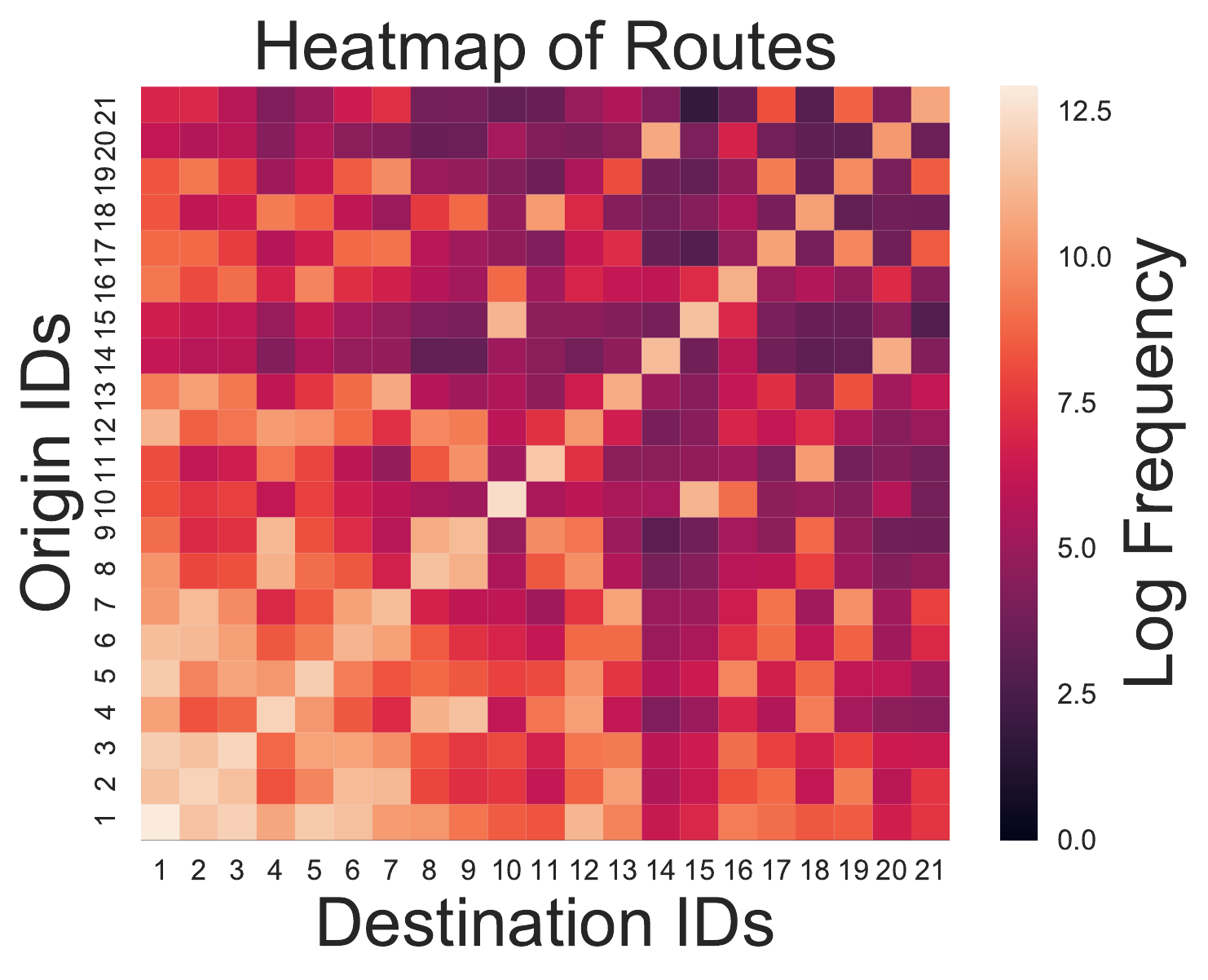}
    \hfill\phantom{1}
		\caption{The logarithmic frequencies of request routes.}
		\label{fig:connectivity}
	\end{figure}

	\begin{figure}
    \hfill
		\subfigure[Time \& price without filtering]{\includegraphics[width=0.3\textwidth]{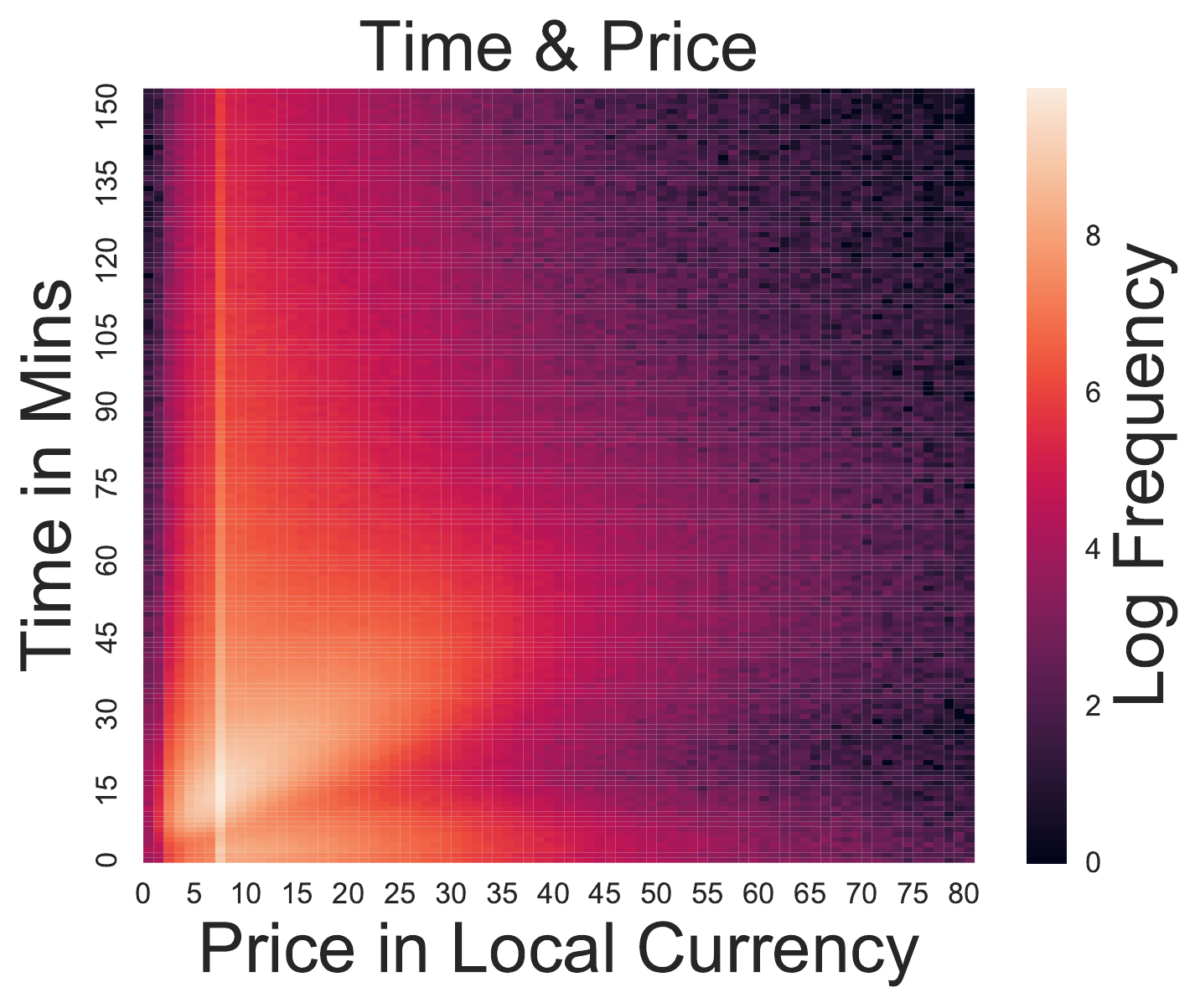}\label{sfig:nofilter}}%
		\hfill %
		\subfigure[Time \& price with filtering]{\includegraphics[width=0.3\textwidth]{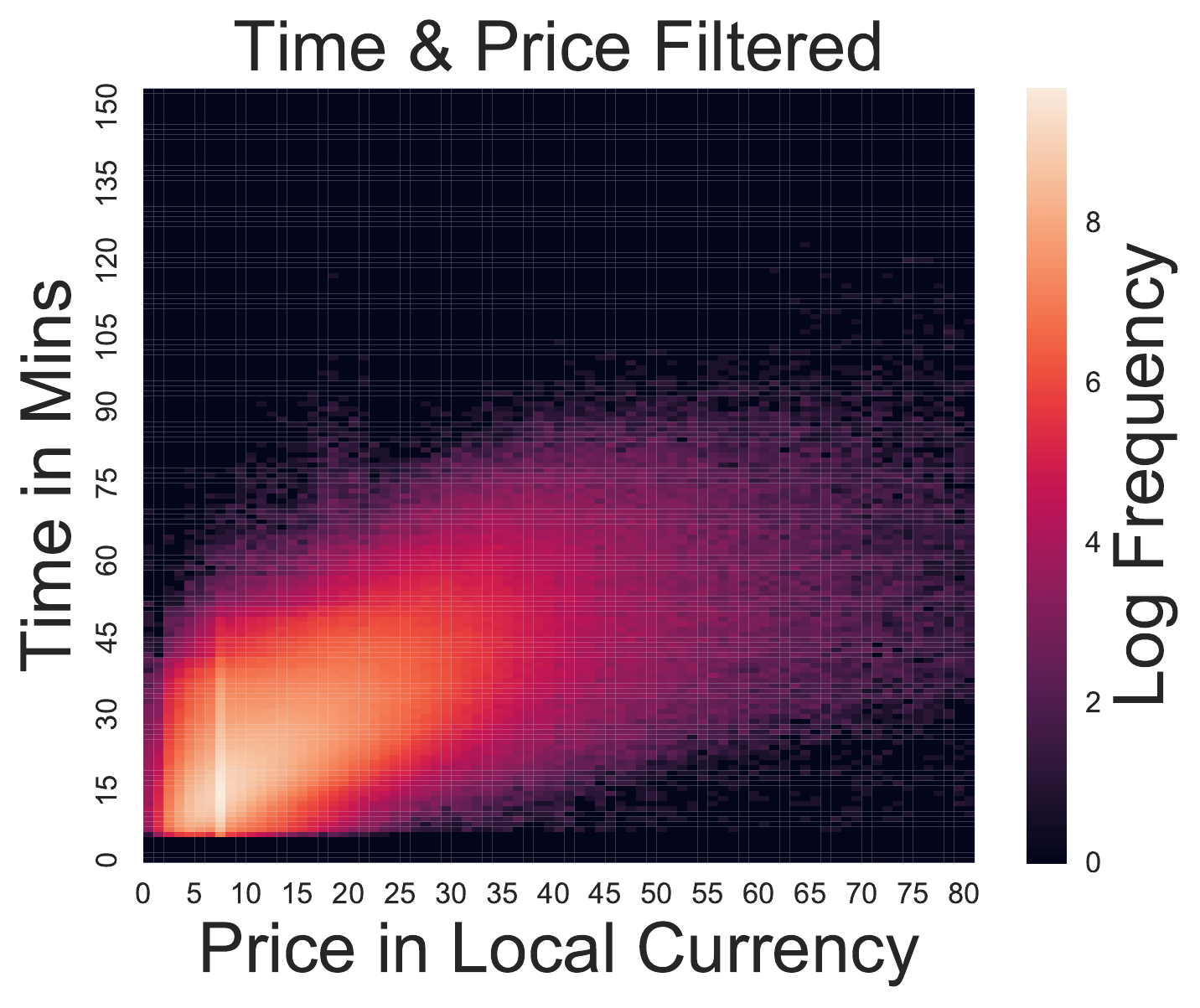}\label{sfig:filter}}
    \hfill\phantom{1}
		\caption{The logarithmic frequencies of $(\text{time},~\text{price})$ pairs, with or
						 without filtering the ``abnormal'' requests.}
		\label{fig:time-price}
	\end{figure}

	\subsection{Data preparation}\label{ssec:clean}
	The time consumptions from nodes to nodes and demand curves for edges are known in our model. However, the dataset doesn't provide such information directly. We filter out "abnormal" requests and apply a linear regression to get the relationship of the travel time and the price. It makes possible to infer the travel time from the order price. For the demand curves, we observe the values of each edge and fit them to lognormal distributions.

	\paragraph{Distance and travel time}
	The distance (or equivalently the
	travel time) from one region to another is required to perform our simulation. We approximate the travel time by
	the time interval of two consecutive requests assigned to the same driver. In
	\autoref{sfig:nofilter}, we plot the frequencies of requests with certain
	(time, price) pairs. We cannot see clear relationship
	between time and price, which are supposed to be roughly linearly
	related in this figure.\footnote{The price of a ride is the maximum of a two-dimension linear
	function of the traveled distance and spent time and a minimal price (which is
	$7$ CNY as one can see the vertical bright line at $\text{price} = 7$ in
	\autoref{fig:time-price}). Since the traveled distance is almost linearly
	related to the spent time, the price, if larger than the minimal price, should
	also be almost linearly related to the traveling time. Readers may notice that
	from the figures, there are many requests with price less than $7$ (even as
	low as $0$). This is because there are many coupons given to passengers to stimulate their demand for riding and the prices given in the dataset are
	after applying the coupons.} We think that this is due to the existence of two types of ``abnormal'' requests:
	\begin{itemize}
		\itemsep0em
		\item {\em Cancelled requests}, usually with very short completion time but
					not necessarily low prices (appeared in the right-bottom part of the
					plot);
		\item {\em The last request of a working period}, after which the driver
					might go home or have a rest. These requests usually have very long completion time
					but not necessarily high enough prices (appeared in the left-top part of the
					plot).
	\end{itemize}

	With the observations above, we filter out the requests with significantly longer
	or shorter travel time compared with most of the requests with the
	same origin and destination. \autoref{sfig:filter} illustrates the frequencies
	of requests after such filtering. As expected, the brightest region roughly
	surrounds the $30^{\circ}$ line in the figure. By applying a standard linear
	regression, the slope turns out to be approximately $0.5117$ CNY per minute.
	One may also notice some ``right-shifting shadows'' of the brightest region,
	which are caused by the surge-pricing policy with different multipliers.

	\begin{figure}
    \hfill
		\includegraphics[width=0.35\textwidth]{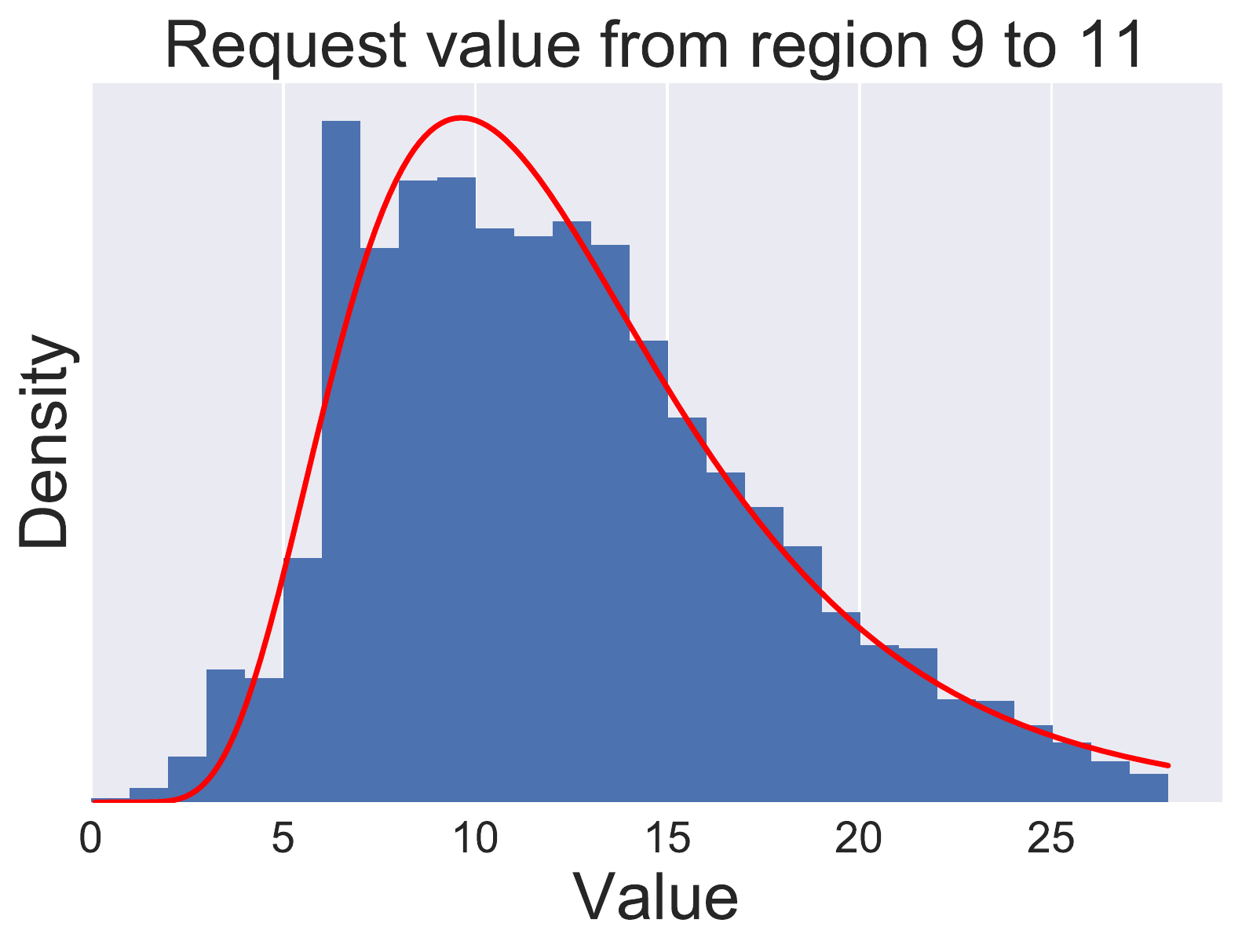}%
		\hfill %
		\includegraphics[width=0.35\textwidth]{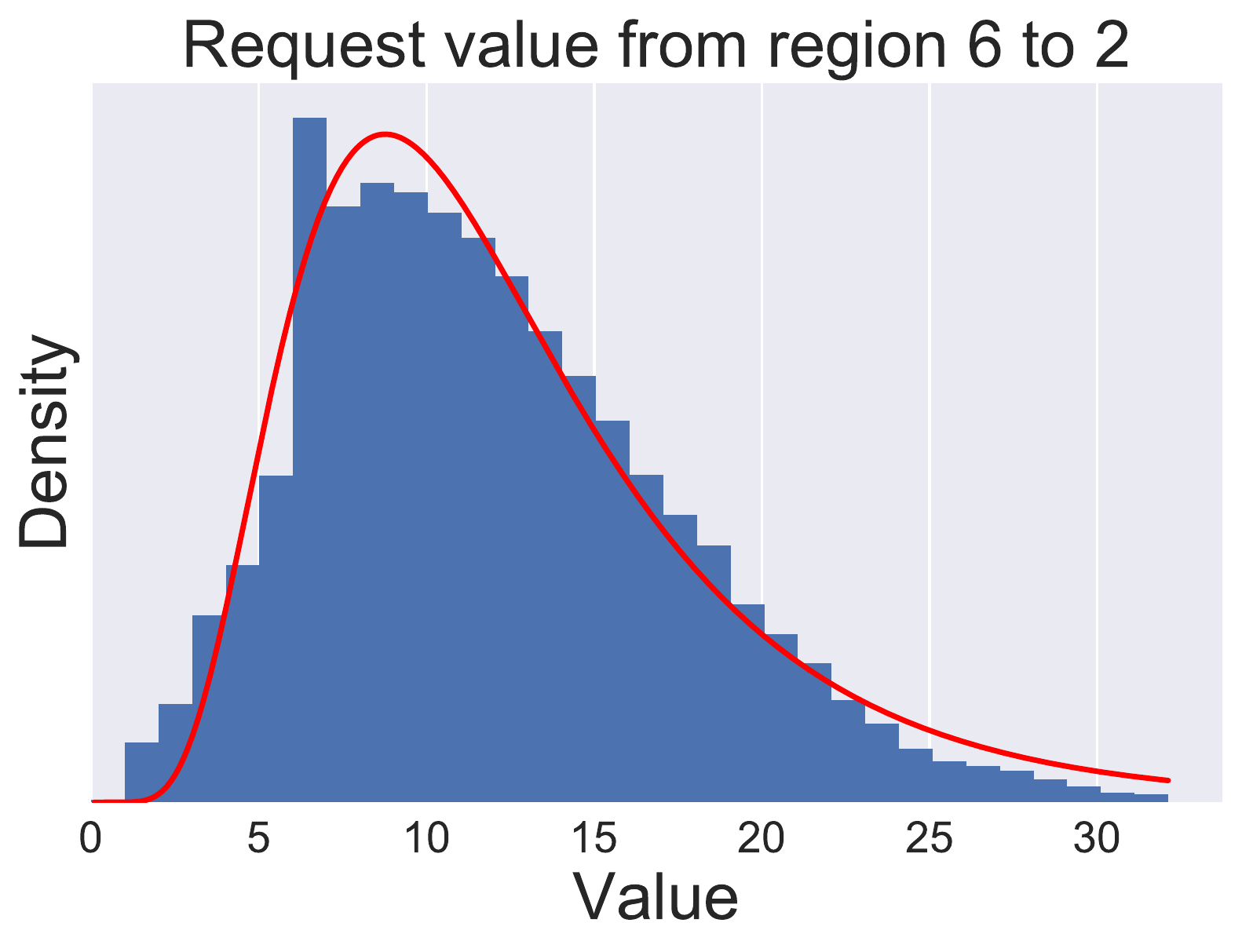}
    \hfill\phantom{1}
		\caption{Fitting request values to lognormal distributions.}\label{fig:lognormal}
	\end{figure}

	\paragraph{Estimation of demand curves}
	To estimate the demand curves, we first gather all the requests along the same
	edge (also within the same time period for dynamic environment, see
	\autoref{ssec:dynamic_env}) and take the prices associated with the requests
	as the values of the passengers. Then, we fit the values of each edge (and each
	time period for dynamic environment) to a lognormal distribution. The reason
	that we choose the lognormal distribution is two-fold: (i) the data fits
	lognormal distributions quite well (see \autoref{fig:lognormal} as examples);
	(ii) lognormal distributions are commonly used in some related literatures
	\cite{ostrovsky2011reserve,lahaie2007revenue,roberts2016ranking,shen2017practical}.

	We set the cost of traveling to be zero, because we do not
	have enough information from the dataset to infer the cost.

	\subsection{Benchmarks}\label{ssec:benchmark}

	We consider two benchmark policies:
	\begin{itemize}
		\item \fixed: fixed per-minute pricing, i.e., the price of a ride equals to
					the estimated traveling time from the origin to the destination of
					this ride multiplied by a per-minute price $\alpha$, where $\alpha$ is
					a constant across the platform.
		\item \surge: based on \fixed policy, using surge pricing to clear the local
					market when supply is not enough. In other words, the price of a ride
					equals to the estimated traveling time multiplied by $\alpha\beta$,
					where $\alpha$ is the fixed per-minute price and $\beta \geq 1$ is the
					surge multiplier. Note that $\beta$ is dynamic and can be different
					for requests initiated at different regions, while the requests
					initiated at the same regions will share the same surge multipliers.
	\end{itemize}

	In the rest of this section, we evaluate and compare our dynamic pricing policy \dynam with
	these two benchmarks in both static and dynamic environments.

	\begin{figure}
    \hfill
		\subfigure[Static environment]{\includegraphics[width=0.35\textwidth]{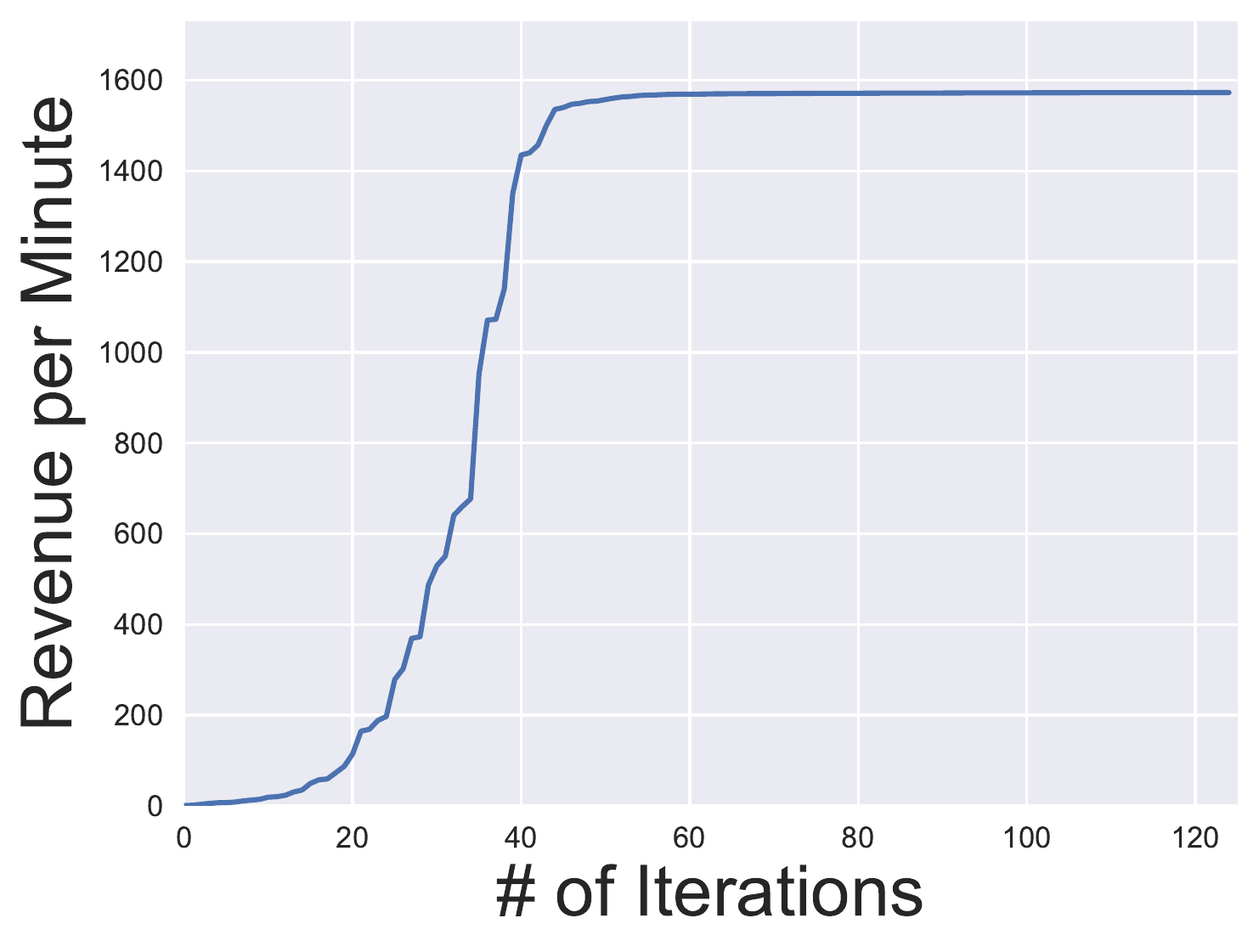}\label{sfig:static-iters}}%
		\hfill %
		\subfigure[Dynamic environment]{\includegraphics[width=0.35\textwidth]{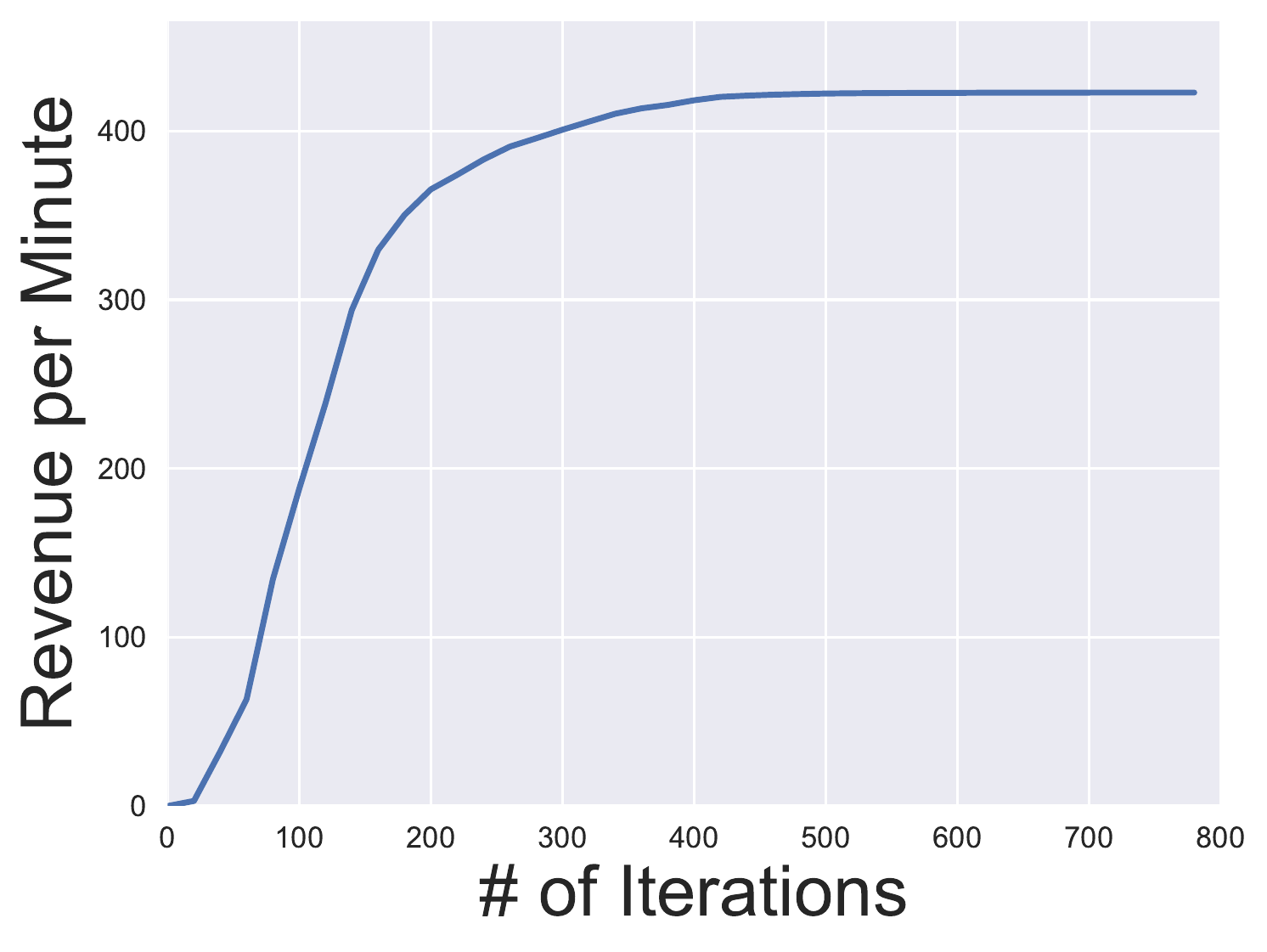}\label{sfig:dynamic-iters}}
    \hfill\phantom{1}
		\caption{Convergence of revenue.}
		\label{fig:converge}
	\end{figure}

	\begin{figure}
    \hfill
		\subfigure[Static environment]{\includegraphics[width=0.35\textwidth]{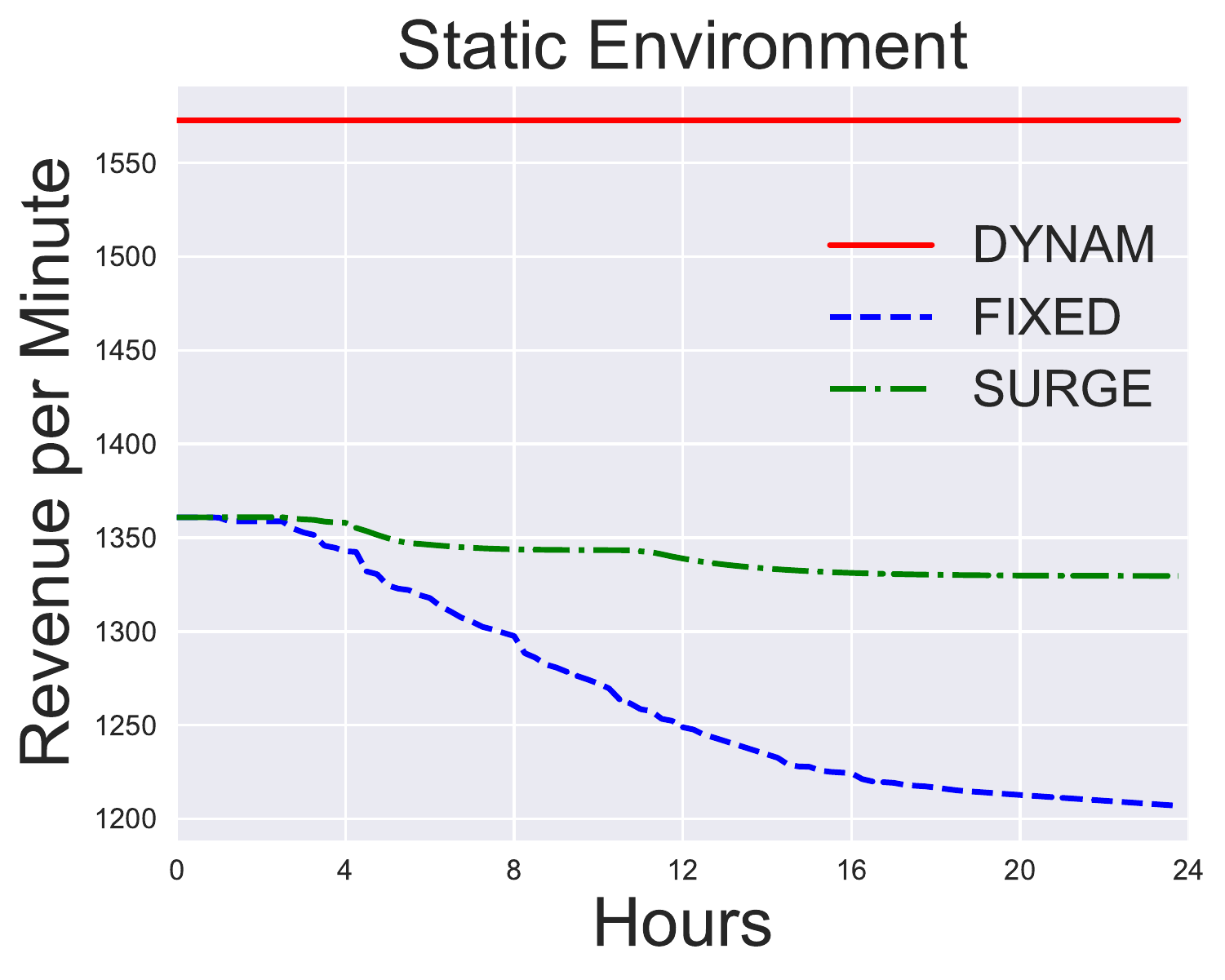}\label{sfig:static-revs}}%
		\hfill %
		\subfigure[Dynamic environment]{\includegraphics[width=0.35\textwidth]{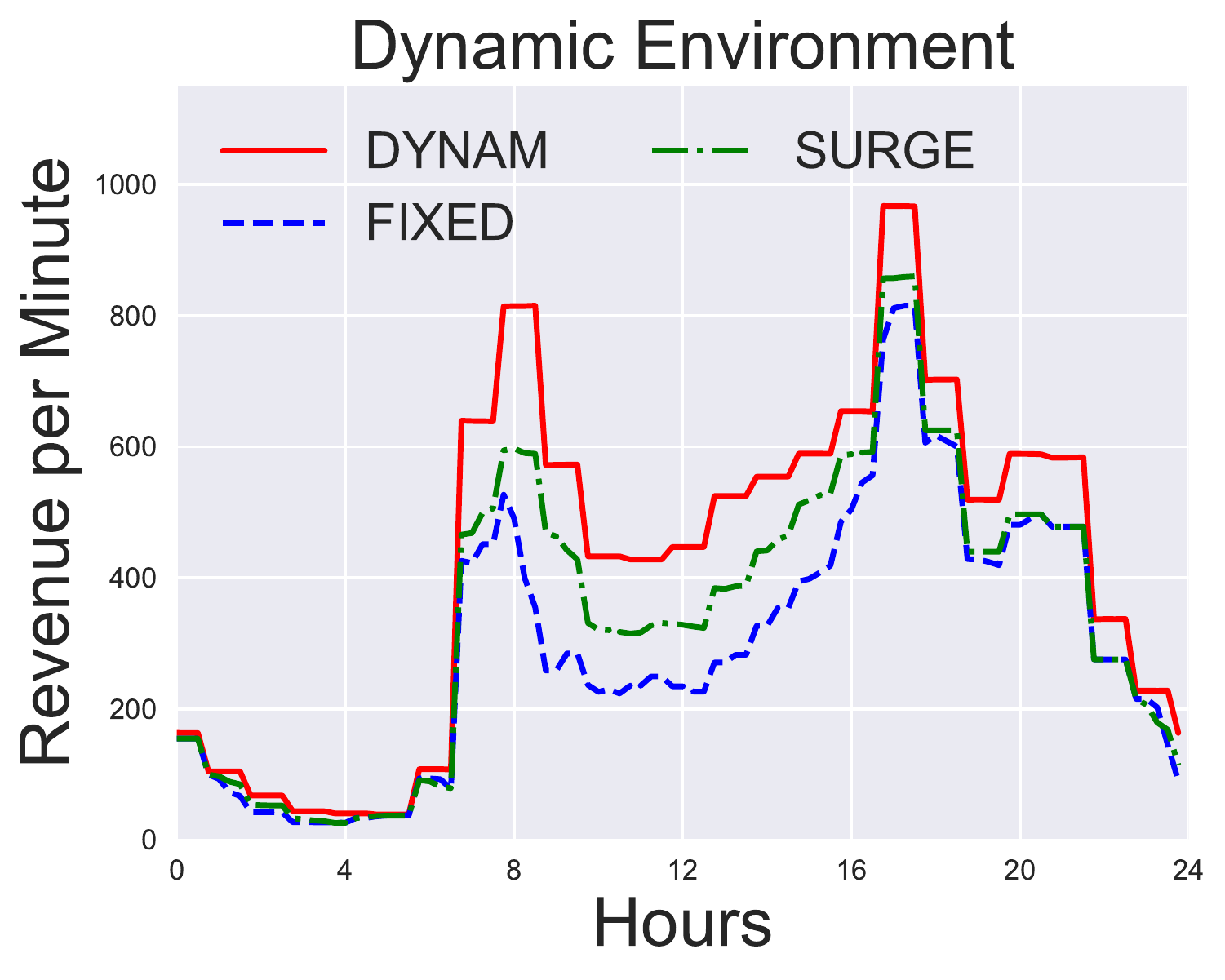}\label{sfig:dynamic-revs}}
    \hfill\phantom{1}
		\caption{Instantaneous revenue in different environments.}
		\label{fig:revs}
	\end{figure}

	\begin{figure*}[t]
		\includegraphics[width=0.1975\textwidth]{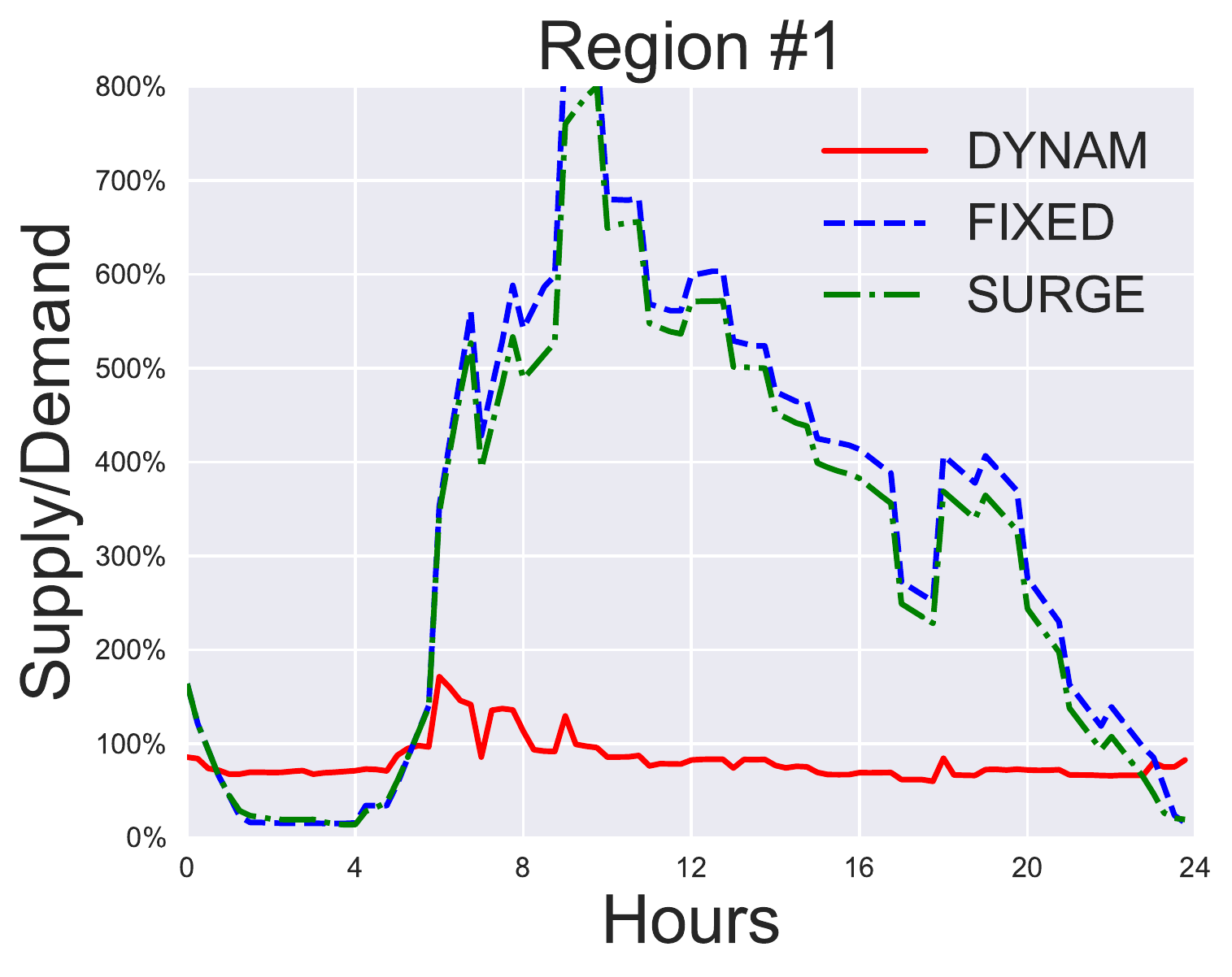}%
		\hfill %
		\includegraphics[width=0.1975\textwidth]{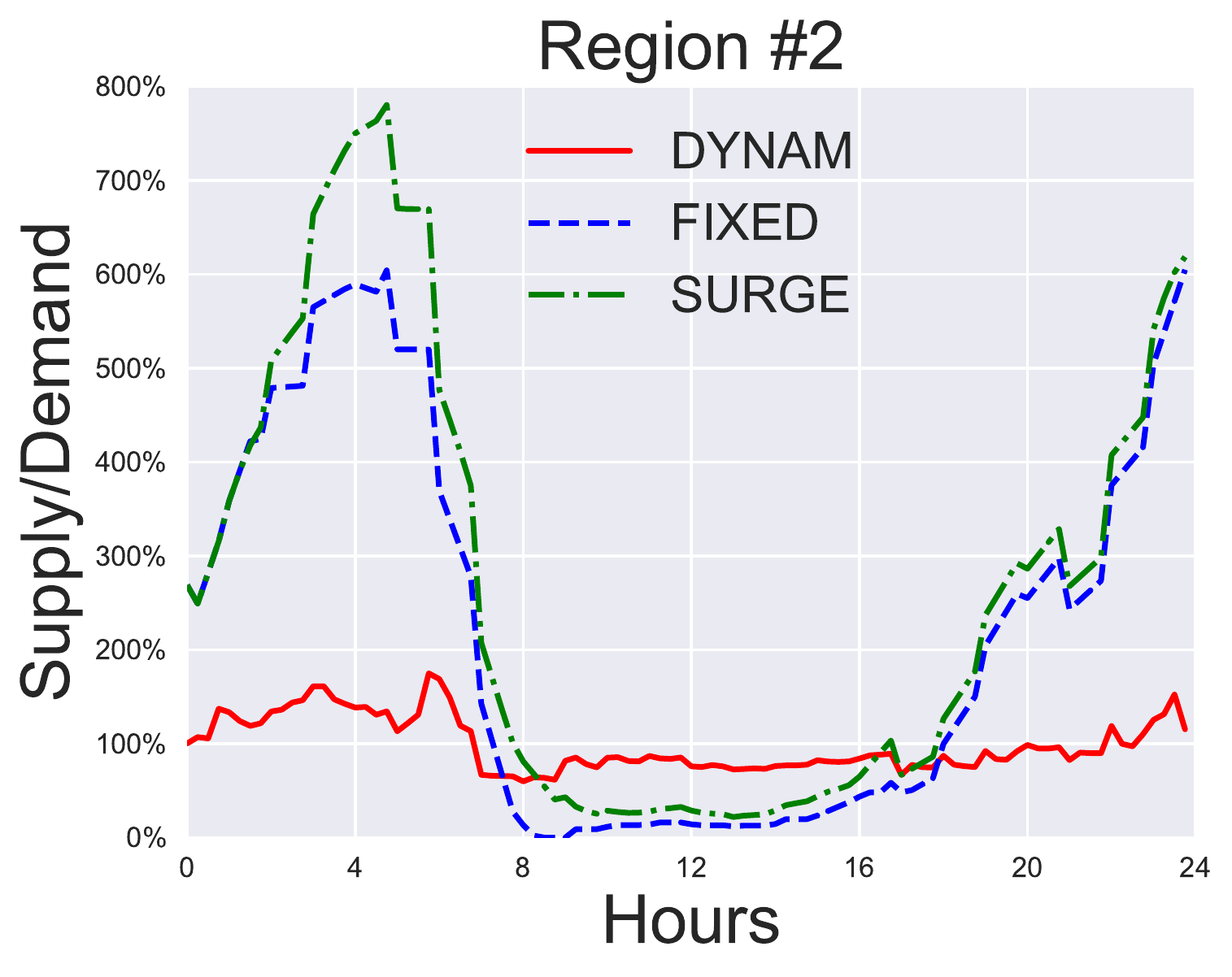}%
		\hfill %
		\includegraphics[width=0.1975\textwidth]{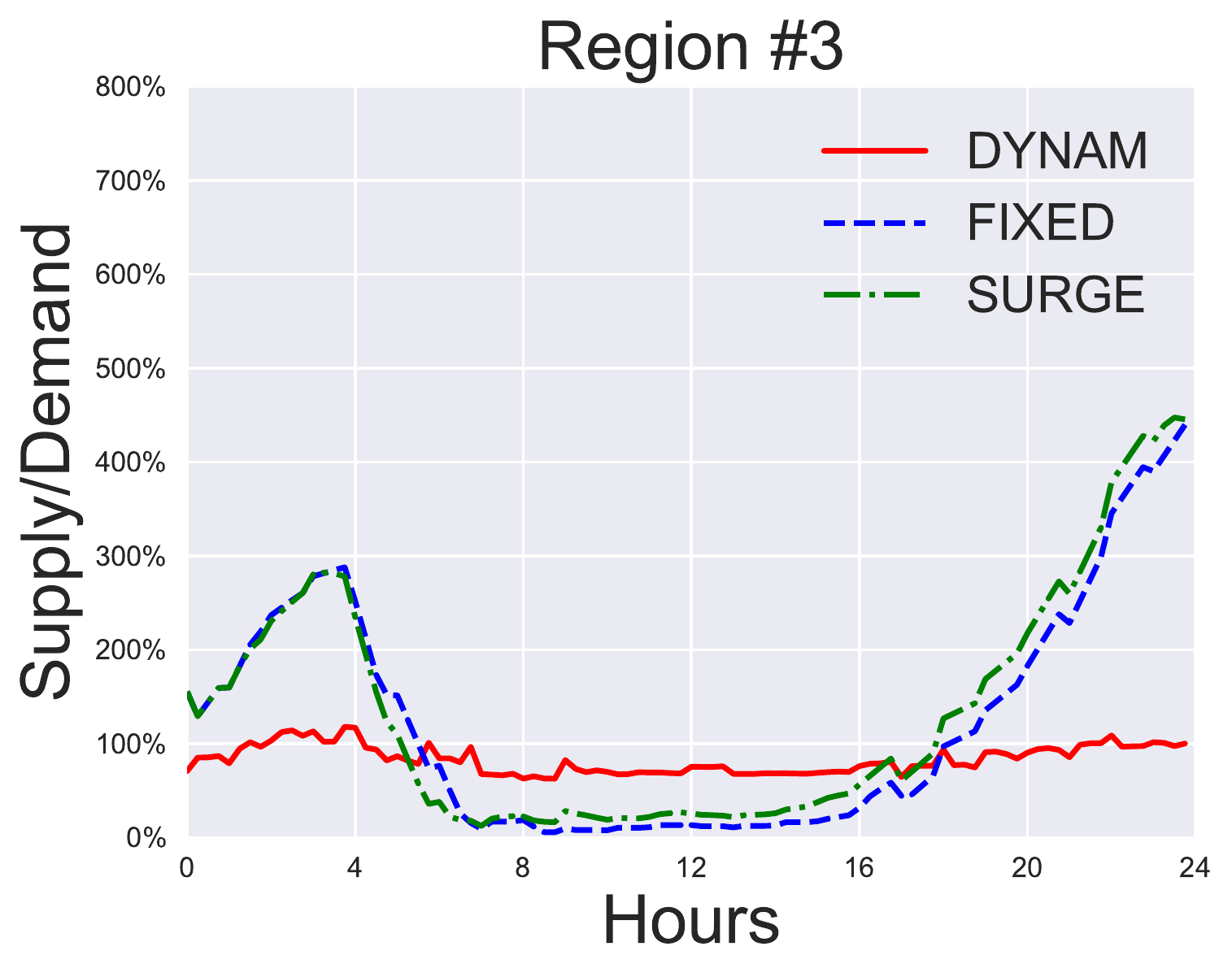}%
		\hfill %
		\includegraphics[width=0.1975\textwidth]{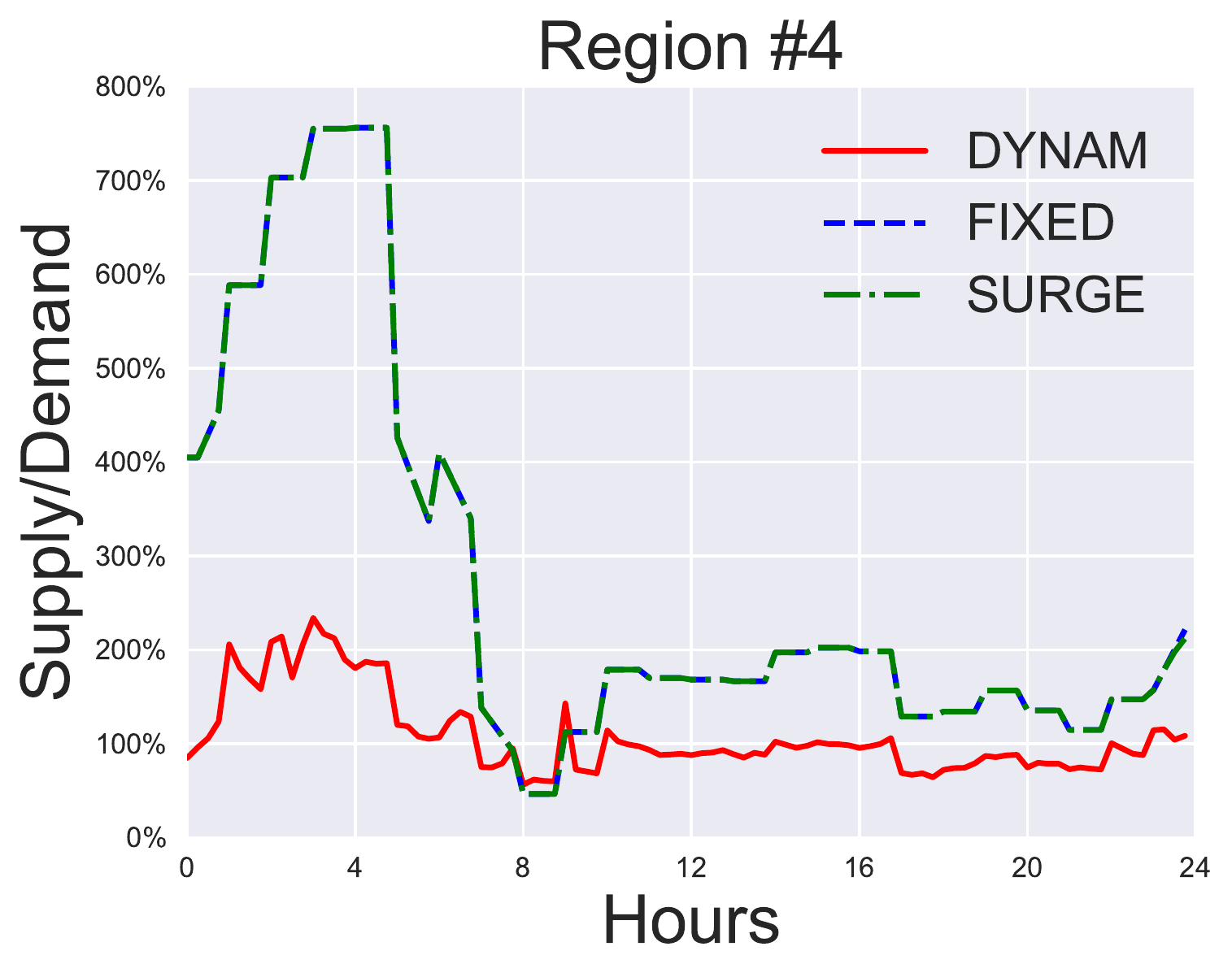}%
		\hfill %
		\includegraphics[width=0.1975\textwidth]{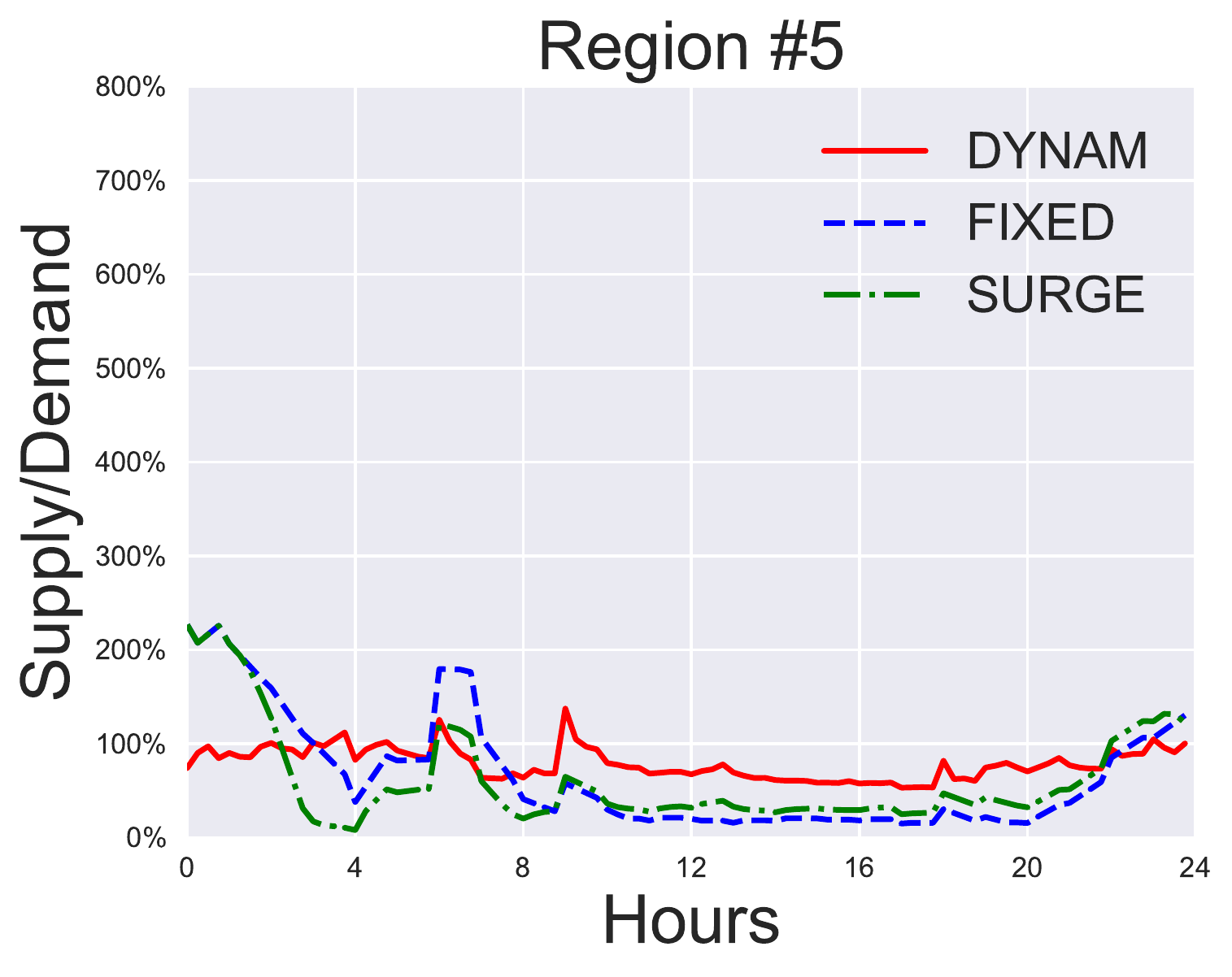}
		\caption{Instantaneous supply ratios for different regions.}
		\label{fig:demand-supply}
	\end{figure*}

\input{static_env.tex}

\input{dynamic_env.tex}

%% file: static_env.tex
\subsection{Static environment}\label{ssec:static_env}

  We first present the empirical analysis for the static environment, which is
  simpler than the dynamic environment that we will consider next, hence easier to
  begin with.

  In the static environment, we use the average of the statistics of all $21$ days as
  the inputs to our model. For example, the demand function $D(p | e)$ is
  estimated based on the frequencies and prices of the requests along edge $e$
  averaged over time. Similarly, the total supply of drivers is estimated based
  on the total durations of completed requests. 

  With the static environment, we can instantiate the convex program
  \eqref{prog:static} and solve via standard gradient descent algorithms. In our
  case, we simply use the MATLAB function \texttt{fmincon} to solve the convex
  program on a PC with Intel i5-3470 CPU. We did not apply any additional
  techniques to speed-up the computation as the optimization of running time is
  not the main focus of this paper. \autoref{sfig:static-iters} illustrates the
  convergence of the objective value (revenue) with increasing number of
  iterations, where each iteration roughly takes $0.2$ second.
  
  To compare the performance of policy \dynam with
  the benchmark policies \fixed and \surge, we also simulates them under the
  same static environment. In particular, the length of each timestep is set to
  be $15$ minutes and the number of steps in simulation is $96$ (so $24$ hours
  in total). For both \fixed and \surge, we use the per-minute price fitted from
  data as the base price, $\alpha = 0.5117$, and allow the surge ratio $\beta$
  to be in $[1.0, 5.0]$. To make the evaluations comparable, we use the
  distribution of drivers under the stationary solution of our convex program as
  the initial driver distributions for \fixed and \surge.
  \autoref{sfig:static-revs} shows how the instantaneous revenues evolve as the
  time goes by, where \dynam on average outperforms \fixed and \surge by roughly $24\%$ and $17\%$, respectively.

  Note that our policy \dynam is stationary under the static environment, the
  instantaneous revenue is constant (the red horizontal line). Interestingly,
  the instantaneous revenue curves of both \fixed and \surge are decreasing and
  the one of \fixed is decreasing much faster. The observation reflects that
  both \fixed and \surge are not doing well in dispatching the vehicles: \fixed
  simply never balances the supply and demand, while \surge shows better control
  in the balance of supply and demand because the policy seeks to balance the
  demand with local supply when supply can not meet the demand. However, neither
  of them really balance the global supply and demand, so the instantaneous
  revenue decrease as the supply and demand become more unbalanced.

  In other words, the empirical analysis supports our insight about the
  importance of vehicle dispatching in ride-sharing platforms.

%% file: dynamic_env.tex
\subsection{Dynamic environment}\label{ssec:dynamic_env}

  In the dynamic environment, the parameters (i.e., the demand functions and the total number of requests) are estimated based on the statistics of {\em each hour}
  but averaged over different days. For example, the demand functions $D_h(p | e)$ are defined for
  each edge $e$ and each of the $24$ hours, $h \in \{0, \ldots, 23\}$. In
  particular, we only use the data from the weekdays ($14$ days in
  total)\footnote{The reason that we only use data from weekdays is that the
  dynamics of demands and supplies in weekdays do have similar patterns but
  quite different from the patterns of weekends.} among the most popular $5$
  regions for the estimation.

  Again, we instantiate the convex program
  \eqref{eq:convex} for the dynamic environment and solve via the \texttt{fmincon}
  function on the same PC that we used for the static case.
  \autoref{sfig:dynamic-iters} shows the convergence of the objective value with
  increasing number of iterations, where each iteration takes less than $1$
  minute.

  We setup \fixed and \surge in exactly the same way as we did for the static
  environment, except that the initial driver distribution is from the solution
  of the convex program for dynamic environment.

  \autoref{sfig:dynamic-revs} shows the instantaneous revenue along the
  simulation. In particular, the relationship $\dynam \succ \surge \succ \fixed$
  holds almost surely. Moreover, the advantages of \dynam over the other two
  policies are more significant at the high-demand ``peak times''.
  For example, at $8$ a.m., \dynam (\texttt{$\sim$}$800$) outperforms \surge
  (\texttt{$\sim$}$600$) and \fixed (\texttt{$\sim$}$500$) by roughly $33\%$
  and $60\%$, respectively.

  \paragraph{Demand-supply balance}
  Balancing the demand and supply is not the goal of our dispatching
  policy. However, a policy without such balancing abilities are unlikely to perform well. In \autoref{fig:demand-supply}, we
  plot the {\em supply ratios} (defined as the local instantaneous supply
  divided by the local instantaneous demand) for all the $5$ regions during the
  $24$ hours of the simulation.

  From the figures, we can easily check that comparing with the other two lines,
  the red line (the supply ratio of \dynam) tightly surrounds the ``balance''
  line of $100\%$, which means that the number of available drivers at any time
  and at each region is close to the number of requests sent from that region at
  that time. The lines of other two policies sometimes could be very far from
  the ``balance'' line, that is, the drivers under policy \fixed and \surge are
  not in the location where many passengers need the service.

  As a result, our policy \dynam shows much stronger power in vehicle dispatching and balancing
  demand and supply in dynamic ride-sharing systems. Such advanced
  techniques in dispatching can in turn help the platform to gain higher revenue
  through serving more passengers.


%% file: main.bbl

\begin{thebibliography}{31}


\ifx \showCODEN    \undefined \def \showCODEN     #1{\unskip}     \fi
\ifx \showDOI      \undefined \def \showDOI       #1{#1}\fi
\ifx \showISBNx    \undefined \def \showISBNx     #1{\unskip}     \fi
\ifx \showISBNxiii \undefined \def \showISBNxiii  #1{\unskip}     \fi
\ifx \showISSN     \undefined \def \showISSN      #1{\unskip}     \fi
\ifx \showLCCN     \undefined \def \showLCCN      #1{\unskip}     \fi
\ifx \shownote     \undefined \def \shownote      #1{#1}          \fi
\ifx \showarticletitle \undefined \def \showarticletitle #1{#1}   \fi
\ifx \showURL      \undefined \def \showURL       {\relax}        \fi
\providecommand\bibfield[2]{#2}
\providecommand\bibinfo[2]{#2}
\providecommand\natexlab[1]{#1}
\providecommand\showeprint[2][]{arXiv:#2}

\bibitem[\protect\citeauthoryear{Alonso-Mora, Samaranayake, Wallar, Frazzoli,
  and Rus}{Alonso-Mora et~al\mbox{.}}{2017}]%
        {alonso2017demand}
\bibfield{author}{\bibinfo{person}{Javier Alonso-Mora},
  \bibinfo{person}{Samitha Samaranayake}, \bibinfo{person}{Alex Wallar},
  \bibinfo{person}{Emilio Frazzoli}, {and} \bibinfo{person}{Daniela Rus}.}
  \bibinfo{year}{2017}\natexlab{}.
\newblock \showarticletitle{On-demand high-capacity ride-sharing via dynamic
  trip-vehicle assignment}.
\newblock \bibinfo{journal}{\emph{PNAS}} (\bibinfo{year}{2017}),
  \bibinfo{pages}{201611675}.
\newblock


\bibitem[\protect\citeauthoryear{Balseiro, Lin, Mirrokni, Paes~Leme, and
  Zuo}{Balseiro et~al\mbox{.}}{2017}]%
        {balseiro2017dynamic}
\bibfield{author}{\bibinfo{person}{Santiago Balseiro}, \bibinfo{person}{Max
  Lin}, \bibinfo{person}{Vahab Mirrokni}, \bibinfo{person}{Renato Paes~Leme},
  {and} \bibinfo{person}{Song Zuo}.} \bibinfo{year}{2017}\natexlab{}.
\newblock \showarticletitle{Dynamic revenue sharing}. In
  \bibinfo{booktitle}{\emph{NIPS 2017}}.
\newblock


\bibitem[\protect\citeauthoryear{Banerjee, Freund, and Lykouris}{Banerjee
  et~al\mbox{.}}{2017}]%
        {banerjee2017pricing}
\bibfield{author}{\bibinfo{person}{Siddhartha Banerjee},
  \bibinfo{person}{Daniel Freund}, {and} \bibinfo{person}{Thodoris Lykouris}.}
  \bibinfo{year}{2017}\natexlab{}.
\newblock \showarticletitle{Pricing and Optimization in Shared Vehicle Systems:
  An Approximation Framework}. In \bibinfo{booktitle}{\emph{EC 2017}}.
\newblock


\bibitem[\protect\citeauthoryear{Banerjee, Riquelme, and Johari}{Banerjee
  et~al\mbox{.}}{2015}]%
        {banerjee2015pricing}
\bibfield{author}{\bibinfo{person}{Siddhartha Banerjee},
  \bibinfo{person}{Carlos Riquelme}, {and} \bibinfo{person}{Ramesh Johari}.}
  \bibinfo{year}{2015}\natexlab{}.
\newblock \showarticletitle{Pricing in Ride-share Platforms: A
  Queueing-Theoretic Approach}.
\newblock  (\bibinfo{year}{2015}).
\newblock


\bibitem[\protect\citeauthoryear{Bimpikis, Candogan, and Daniela}{Bimpikis
  et~al\mbox{.}}{2016}]%
        {bimpikis2016spatial}
\bibfield{author}{\bibinfo{person}{Kostas Bimpikis}, \bibinfo{person}{Ozan
  Candogan}, {and} \bibinfo{person}{Saban Daniela}.}
  \bibinfo{year}{2016}\natexlab{}.
\newblock \showarticletitle{Spatial Pricing in Ride-Sharing Networks}.
\newblock  (\bibinfo{year}{2016}).
\newblock


\bibitem[\protect\citeauthoryear{Boyd and Vandenberghe}{Boyd and
  Vandenberghe}{2004}]%
        {boyd2004convex}
\bibfield{author}{\bibinfo{person}{Stephen Boyd} {and} \bibinfo{person}{Lieven
  Vandenberghe}.} \bibinfo{year}{2004}\natexlab{}.
\newblock \bibinfo{booktitle}{\emph{Convex optimization}}.
\newblock \bibinfo{publisher}{Cambridge university press}.
\newblock


\bibitem[\protect\citeauthoryear{Cachon, Daniels, and Lobel}{Cachon
  et~al\mbox{.}}{2016}]%
        {cachon2016role}
\bibfield{author}{\bibinfo{person}{Gerard~P Cachon}, \bibinfo{person}{Kaitlin~M
  Daniels}, {and} \bibinfo{person}{Ruben Lobel}.}
  \bibinfo{year}{2016}\natexlab{}.
\newblock \showarticletitle{The role of surge pricing on a service platform
  with self-scheduling capacity}.
\newblock  (\bibinfo{year}{2016}).
\newblock


\bibitem[\protect\citeauthoryear{Castillo, Knoepfle, and Weyl}{Castillo
  et~al\mbox{.}}{2017}]%
        {castillo2017surge}
\bibfield{author}{\bibinfo{person}{Juan~Camilo Castillo}, \bibinfo{person}{Dan
  Knoepfle}, {and} \bibinfo{person}{Glen Weyl}.}
  \bibinfo{year}{2017}\natexlab{}.
\newblock \showarticletitle{Surge pricing solves the wild goose chase}. In
  \bibinfo{booktitle}{\emph{EC 2017}}. ACM, \bibinfo{pages}{241--242}.
\newblock


\bibitem[\protect\citeauthoryear{Chan and Shaheen}{Chan and Shaheen}{2012}]%
        {chan2012ridesharing}
\bibfield{author}{\bibinfo{person}{Nelson~D Chan} {and}
  \bibinfo{person}{Susan~A Shaheen}.} \bibinfo{year}{2012}\natexlab{}.
\newblock \showarticletitle{Ridesharing in north america: Past, present, and
  future}.
\newblock \bibinfo{journal}{\emph{Transport Reviews}} \bibinfo{volume}{32},
  \bibinfo{number}{1} (\bibinfo{year}{2012}), \bibinfo{pages}{93--112}.
\newblock


\bibitem[\protect\citeauthoryear{Chen and Sheldon}{Chen and Sheldon}{2015}]%
        {chen2015dynamic}
\bibfield{author}{\bibinfo{person}{M~Keith Chen} {and} \bibinfo{person}{Michael
  Sheldon}.} \bibinfo{year}{2015}\natexlab{}.
\newblock \bibinfo{booktitle}{\emph{Dynamic pricing in a labor market: Surge
  pricing and flexible work on the Uber platform}}.
\newblock \bibinfo{type}{{T}echnical {R}eport}. \bibinfo{institution}{Mimeo,
  UCLA}.
\newblock


\bibitem[\protect\citeauthoryear{Cramer and Krueger}{Cramer and
  Krueger}{2016}]%
        {cramer2016disruptive}
\bibfield{author}{\bibinfo{person}{Judd Cramer} {and} \bibinfo{person}{Alan~B
  Krueger}.} \bibinfo{year}{2016}\natexlab{}.
\newblock \showarticletitle{Disruptive change in the taxi business: The case of
  Uber}.
\newblock \bibinfo{journal}{\emph{The American Economic Review}}
  \bibinfo{volume}{106}, \bibinfo{number}{5} (\bibinfo{year}{2016}),
  \bibinfo{pages}{177--182}.
\newblock


\bibitem[\protect\citeauthoryear{Crawford and Meng}{Crawford and Meng}{2011}]%
        {crawford2011new}
\bibfield{author}{\bibinfo{person}{Vincent~P Crawford} {and}
  \bibinfo{person}{Juanjuan Meng}.} \bibinfo{year}{2011}\natexlab{}.
\newblock \showarticletitle{New york city cab drivers' labor supply revisited:
  Reference-dependent preferences with rationalexpectations targets for hours
  and income}.
\newblock \bibinfo{journal}{\emph{AER}} \bibinfo{volume}{101},
  \bibinfo{number}{5} (\bibinfo{year}{2011}), \bibinfo{pages}{1912--1932}.
\newblock


\bibitem[\protect\citeauthoryear{Fang, Huang, and Wierman}{Fang
  et~al\mbox{.}}{2017}]%
        {fang2017prices}
\bibfield{author}{\bibinfo{person}{Zhixuan Fang}, \bibinfo{person}{Longbo
  Huang}, {and} \bibinfo{person}{Adam Wierman}.}
  \bibinfo{year}{2017}\natexlab{}.
\newblock \showarticletitle{Prices and subsidies in the sharing economy}. In
  \bibinfo{booktitle}{\emph{Proceedings of the 26th International Conference on
  World Wide Web}}. WWW 2017, \bibinfo{pages}{53--62}.
\newblock


\bibitem[\protect\citeauthoryear{Gale and Holmes}{Gale and Holmes}{1993}]%
        {gale1993advance}
\bibfield{author}{\bibinfo{person}{Ian~L Gale} {and} \bibinfo{person}{Thomas~J
  Holmes}.} \bibinfo{year}{1993}\natexlab{}.
\newblock \showarticletitle{Advance-purchase discounts and monopoly allocation
  of capacity}.
\newblock \bibinfo{journal}{\emph{The American Economic Review}}
  (\bibinfo{year}{1993}), \bibinfo{pages}{135--146}.
\newblock


\bibitem[\protect\citeauthoryear{Gendreau, Hertz, and Laporte}{Gendreau
  et~al\mbox{.}}{1994}]%
        {gendreau1994tabu}
\bibfield{author}{\bibinfo{person}{Michel Gendreau}, \bibinfo{person}{Alain
  Hertz}, {and} \bibinfo{person}{Gilbert Laporte}.}
  \bibinfo{year}{1994}\natexlab{}.
\newblock \showarticletitle{A tabu search heuristic for the vehicle routing
  problem}.
\newblock \bibinfo{journal}{\emph{Management science}} \bibinfo{volume}{40},
  \bibinfo{number}{10} (\bibinfo{year}{1994}), \bibinfo{pages}{1276--1290}.
\newblock


\bibitem[\protect\citeauthoryear{Ghiani, Guerriero, Laporte, and
  Musmanno}{Ghiani et~al\mbox{.}}{2003}]%
        {ghiani2003real}
\bibfield{author}{\bibinfo{person}{Gianpaolo Ghiani},
  \bibinfo{person}{Francesca Guerriero}, \bibinfo{person}{Gilbert Laporte},
  {and} \bibinfo{person}{Roberto Musmanno}.} \bibinfo{year}{2003}\natexlab{}.
\newblock \showarticletitle{Real-time vehicle routing: Solution concepts,
  algorithms and parallel computing strategies}.
\newblock \bibinfo{journal}{\emph{European Journal of Operational Research}}
  \bibinfo{volume}{151}, \bibinfo{number}{1} (\bibinfo{year}{2003}).
\newblock


\bibitem[\protect\citeauthoryear{Iglesias, Rossi, Zhang, and Pavone}{Iglesias
  et~al\mbox{.}}{2016}]%
        {iglesias2016bcmp}
\bibfield{author}{\bibinfo{person}{Ramon Iglesias}, \bibinfo{person}{Federico
  Rossi}, \bibinfo{person}{Rick Zhang}, {and} \bibinfo{person}{Marco Pavone}.}
  \bibinfo{year}{2016}\natexlab{}.
\newblock \showarticletitle{A BCMP Network Approach to Modeling and Controlling
  Autonomous Mobility-on-Demand Systems}.
\newblock \bibinfo{journal}{\emph{arXiv preprint arXiv:1607.04357}}
  (\bibinfo{year}{2016}).
\newblock


\bibitem[\protect\citeauthoryear{Kostiuk}{Kostiuk}{1990}]%
        {kostiuk1990compensating}
\bibfield{author}{\bibinfo{person}{Peter~F Kostiuk}.}
  \bibinfo{year}{1990}\natexlab{}.
\newblock \showarticletitle{Compensating differentials for shift work}.
\newblock \bibinfo{journal}{\emph{Journal of political Economy}}
  \bibinfo{volume}{98}, \bibinfo{number}{5, Part 1} (\bibinfo{year}{1990}),
  \bibinfo{pages}{1054--1075}.
\newblock


\bibitem[\protect\citeauthoryear{Lahaie and Pennock}{Lahaie and
  Pennock}{2007}]%
        {lahaie2007revenue}
\bibfield{author}{\bibinfo{person}{S{\'e}bastien Lahaie} {and}
  \bibinfo{person}{David~M Pennock}.} \bibinfo{year}{2007}\natexlab{}.
\newblock \showarticletitle{Revenue analysis of a family of ranking rules for
  keyword auctions}. In \bibinfo{booktitle}{\emph{EC 2007}}. ACM,
  \bibinfo{pages}{50--56}.
\newblock


\bibitem[\protect\citeauthoryear{Laporte}{Laporte}{1992}]%
        {laporte1992vehicle}
\bibfield{author}{\bibinfo{person}{Gilbert Laporte}.}
  \bibinfo{year}{1992}\natexlab{}.
\newblock \showarticletitle{The vehicle routing problem: An overview of exact
  and approximate algorithms}.
\newblock \bibinfo{journal}{\emph{European journal of operational research}}
  \bibinfo{volume}{59}, \bibinfo{number}{3} (\bibinfo{year}{1992}).
\newblock


\bibitem[\protect\citeauthoryear{Ma, Zheng, and Wolfson}{Ma
  et~al\mbox{.}}{2013}]%
        {ma2013t}
\bibfield{author}{\bibinfo{person}{Shuo Ma}, \bibinfo{person}{Yu Zheng}, {and}
  \bibinfo{person}{Ouri Wolfson}.} \bibinfo{year}{2013}\natexlab{}.
\newblock \showarticletitle{T-share: A large-scale dynamic taxi ridesharing
  service}. In \bibinfo{booktitle}{\emph{ICDE}}. IEEE,
  \bibinfo{pages}{410--421}.
\newblock


\bibitem[\protect\citeauthoryear{McAfee and Te~Velde}{McAfee and
  Te~Velde}{2006}]%
        {mcafee2006dynamic}
\bibfield{author}{\bibinfo{person}{R~Preston McAfee} {and}
  \bibinfo{person}{Vera Te~Velde}.} \bibinfo{year}{2006}\natexlab{}.
\newblock \showarticletitle{Dynamic pricing in the airline industry}.
\newblock \bibinfo{journal}{\emph{forthcoming in Handbook on Economics and
  Information Systems, Ed: TJ Hendershott, Elsevier}} (\bibinfo{year}{2006}).
\newblock


\bibitem[\protect\citeauthoryear{Moreira-Matias, Gama, Ferreira,
  Mendes-Moreira, and Damas}{Moreira-Matias et~al\mbox{.}}{2013}]%
        {moreira2013predicting}
\bibfield{author}{\bibinfo{person}{Luis Moreira-Matias}, \bibinfo{person}{Joao
  Gama}, \bibinfo{person}{Michel Ferreira}, \bibinfo{person}{Joao
  Mendes-Moreira}, {and} \bibinfo{person}{Luis Damas}.}
  \bibinfo{year}{2013}\natexlab{}.
\newblock \showarticletitle{Predicting taxi--passenger demand using streaming
  data}.
\newblock \bibinfo{journal}{\emph{IEEE Transactions on Intelligent
  Transportation Systems}} \bibinfo{volume}{14}, \bibinfo{number}{3}
  (\bibinfo{year}{2013}), \bibinfo{pages}{1393--1402}.
\newblock


\bibitem[\protect\citeauthoryear{Oettinger}{Oettinger}{1999}]%
        {oettinger1999empirical}
\bibfield{author}{\bibinfo{person}{Gerald~S Oettinger}.}
  \bibinfo{year}{1999}\natexlab{}.
\newblock \showarticletitle{An empirical analysis of the daily labor supply of
  stadium venors}.
\newblock \bibinfo{journal}{\emph{Journal of political Economy}}
  \bibinfo{volume}{107}, \bibinfo{number}{2} (\bibinfo{year}{1999}),
  \bibinfo{pages}{360--392}.
\newblock


\bibitem[\protect\citeauthoryear{Ostrovsky and Schwarz}{Ostrovsky and
  Schwarz}{2011}]%
        {ostrovsky2011reserve}
\bibfield{author}{\bibinfo{person}{Michael Ostrovsky} {and}
  \bibinfo{person}{Michael Schwarz}.} \bibinfo{year}{2011}\natexlab{}.
\newblock \showarticletitle{Reserve prices in internet advertising auctions: A
  field experiment}. In \bibinfo{booktitle}{\emph{EC 2011Practical}}. ACM,
  \bibinfo{pages}{59--60}.
\newblock


\bibitem[\protect\citeauthoryear{Roberts, Gunawardena, Kash, and Key}{Roberts
  et~al\mbox{.}}{2016}]%
        {roberts2016ranking}
\bibfield{author}{\bibinfo{person}{Ben Roberts}, \bibinfo{person}{Dinan
  Gunawardena}, \bibinfo{person}{Ian~A Kash}, {and} \bibinfo{person}{Peter
  Key}.} \bibinfo{year}{2016}\natexlab{}.
\newblock \showarticletitle{Ranking and tradeoffs in sponsored search
  auctions}.
\newblock \bibinfo{journal}{\emph{ACM Transactions on Economics and
  Computation}} \bibinfo{volume}{4}, \bibinfo{number}{3}
  (\bibinfo{year}{2016}), \bibinfo{pages}{17}.
\newblock


\bibitem[\protect\citeauthoryear{Shen and Tang}{Shen and Tang}{2017}]%
        {shen2017practical}
\bibfield{author}{\bibinfo{person}{Weiran Shen} {and}
  \bibinfo{person}{Pingzhong Tang}.} \bibinfo{year}{2017}\natexlab{}.
\newblock \showarticletitle{Practical versus Optimal Mechanisms}. In
  \bibinfo{booktitle}{\emph{AAMAS}}. \bibinfo{pages}{78--86}.
\newblock


\bibitem[\protect\citeauthoryear{Stavins}{Stavins}{2001}]%
        {stavins2001price}
\bibfield{author}{\bibinfo{person}{Joanna Stavins}.}
  \bibinfo{year}{2001}\natexlab{}.
\newblock \showarticletitle{Price discrimination in the airline market: The
  effect of market concentration}.
\newblock \bibinfo{journal}{\emph{Review of Economics and Statistics}}
  \bibinfo{volume}{83}, \bibinfo{number}{1} (\bibinfo{year}{2001}),
  \bibinfo{pages}{200--202}.
\newblock


\bibitem[\protect\citeauthoryear{Tang, Bai, So, Chen, and Wang}{Tang
  et~al\mbox{.}}{2016}]%
        {tang2016coordinating}
\bibfield{author}{\bibinfo{person}{Christopher~S Tang}, \bibinfo{person}{Jiaru
  Bai}, \bibinfo{person}{Kut~C So}, \bibinfo{person}{Xiqun~Michael Chen}, {and}
  \bibinfo{person}{Hai Wang}.} \bibinfo{year}{2016}\natexlab{}.
\newblock \showarticletitle{Coordinating supply and demand on an on-demand
  platform: Price, wage, and payout ratio}.
\newblock  (\bibinfo{year}{2016}).
\newblock


\bibitem[\protect\citeauthoryear{Tong, Chen, Zhou, Chen, Wang, Yang, Ye, and
  Lv}{Tong et~al\mbox{.}}{2017}]%
        {tong2017simpler}
\bibfield{author}{\bibinfo{person}{Yongxin Tong}, \bibinfo{person}{Yuqiang
  Chen}, \bibinfo{person}{Zimu Zhou}, \bibinfo{person}{Lei Chen},
  \bibinfo{person}{Jie Wang}, \bibinfo{person}{Qiang Yang},
  \bibinfo{person}{Jieping Ye}, {and} \bibinfo{person}{Weifeng Lv}.}
  \bibinfo{year}{2017}\natexlab{}.
\newblock \showarticletitle{The simpler the better: a unified approach to
  predicting original taxi demands based on large-scale online platforms}. In
  \bibinfo{booktitle}{\emph{KDD 2017}}. ACM, \bibinfo{pages}{1653--1662}.
\newblock


\bibitem[\protect\citeauthoryear{Zhao, Zhang, Gerding, Sakurai, and Yokoo}{Zhao
  et~al\mbox{.}}{2014}]%
        {zhao2014incentives}
\bibfield{author}{\bibinfo{person}{Dengji Zhao}, \bibinfo{person}{Dongmo
  Zhang}, \bibinfo{person}{Enrico~H Gerding}, \bibinfo{person}{Yuko Sakurai},
  {and} \bibinfo{person}{Makoto Yokoo}.} \bibinfo{year}{2014}\natexlab{}.
\newblock \showarticletitle{Incentives in ridesharing with deficit control}. In
  \bibinfo{booktitle}{\emph{AAMAS 2014}}.
\newblock


\end{thebibliography}
